\documentclass[AMA,STIX1COL]{WileyNJD-v2}
\usepackage{subfigure}
\usepackage{tikz}
\usetikzlibrary{decorations.pathmorphing}
\usetikzlibrary{decorations.pathreplacing}
\date{}
\renewcommand{\thefootnote}{\fnsymbol{footnote}}

\articletype{Research Article}

\received{}
\revised{}
\accepted{}

\begin{document}
\title{{Limiting Spectral Distribution of High-dimensional Hayashi-Yoshida Estimator of Integrated Covariance Matrix}}
\author[1]{Arnab Chakrabarti}
\author[2]{Rituparna Sen*}

\authormark{Chakrabarti and Sen}

\address[1]{\orgdiv{Misra Centre for Financial Markets and Economy}, \orgname{Indian Institute of Management}, \orgaddress{Ahmedabad, \state{GJ}, \country{India}}}
\address[2]{\orgdiv{Applied Statistics Division}, \orgname{Indian Statistical Institute}, \orgaddress{Bangalore, \state{KA}, \country{India}}}

\corres{*Rituparna Sen, Indian Statistical Institute, Bangalore, KA, India. \email{ritupar.sen@gmail.com}}

\abstract[Summary]{In this paper, the estimation of the Integrated Covariance matrix from high-frequency data, for high dimensional stock price process, is considered.   The Hayashi-Yoshida covolatility estimator is an improvement over Realized covolatility for asynchronous data and works well in low dimensions. However it becomes inconsistent and unreliable in the high dimensional situation. We study the bulk spectrum of this matrix and establish its connection to the spectrum of the true covariance matrix in the limiting case where the dimension goes to infinity. The results are illustrated with simulation studies in finite, but high, dimensional cases. An application to real data with tick-by-tick data on 50 stocks is presented. }

\keywords{Asynchronicity, Integrated Covariance, Realized Variance, Spectral Distribution, High-frequency data. }

\maketitle
\footnotetext{\textbf{Acknowledgement:}The authors thank Vikram Sarabhai library for
the help with obtaining the data.} 
\footnotetext{\textbf{Codes:} All the R codes are available in the following GitHub repository https://github.com/Arnabchakrabarti15/LSD-of-Hayashi-Yoshida-estimator.}
\section{Introduction}\label{sec:Introduction}

Intraday financial data of multiple stocks are almost always nonsynchronous in nature. If not adjusted appropriately, nonsynchronous trading can affect multivariate stock price data analysis and the resulting inference quite heavily \cite{epps1979comovements}. The analysis of (intraday) financial data would fail to capture the reality of the financial market if the effect of asynchronicity is ignored. \cite{baumohl2010stock}. Despite this problem, intraday data are important to measure the (co)variance of daily {\it log} returns of a given set of securities and can offer additional information compared to the estimate obtained from daily financial data. This covariance is called {\it integrated (co)variance} (or {\it integrated (co)volatility}). For low-dimensional stock-price process, integrated covariance can be accurately estimated by the {\it Hayashi-Yoshida estimator} \cite{hayashi2005covariance}. But in case of high-dimensional data, it suffers from the same problems as any sample covariance matrix under high-dimensional set up ~\cite{pourahmadi2013high}. As a consequence, its eigenvalue spectrum deviates considerably from the population counterpart. In this study, we derive the limiting spectral distribution of the Hayashi-Yoshida estimator for high-dimensional data.

In 1979, T. W. Epps reported that stock return correlations decrease
as the sampling frequency of data increases \cite{epps1979comovements}. This is one of the earliest manifestations of the problems caused by asynchronicity and is known as the \emph{Epps effect}. Later the phenomenon
has been reported in several studies of different stock markets \cite{tumminello2007correlation,zebedee2009closer,bonanno2001high}
and foreign exchange markets \cite{muthuswamy2001time,lundin1998correlation}.
This is primarily a result of asynchronicity of price observations
and the existing lead-lag relation between asset prices \cite{reno2003closer,precup2004comparison,lo1990econometric}. Empirical results showed that considering only the synchronous or nearly synchronous ticks mitigates the problem significantly \cite{reno2003closer}. In several studies it was shown that asynchronicity can induce potentially serious bias in the estimates of moments and co-moments of asset returns, such as their means, variances, covariances, betas, and autocorrelation and cross-autocorrelation coefficients \cite{lo1990econometric, campbell1997econometrics,bernhardt2008impact,atchison1987nonsynchronous}.

Integrated volatility is defined as the variance of the log return over a day of a given security. In fact, due to its high frequency, intraday financial data has been proven to be more efficient in measuring daily volatility when compared to daily financial data \cite{barndorff2004econometric}. For a single stock, Merton (1980) showed that the variance over a fixed interval can be estimated accurately as sum of squared realization as long as the data are available at sufficiently high sampling frequency \cite{merton1980estimating}. This estimator is known as \textit{Realized volatility}. But often univariate modeling is not sufficient, it is also important to model the correlation dynamics between several assets.  Hence one of the parameters of interest, to accurately estimate and infer about, is the  \textit{integrated co-volatility} or \textit{integrated covariance} matrix.  Analogous to the realized volatility, for a multivariate stock price process, the \textit{realized covolatility matrix}/\textit{realized covariance matrix} is defined. But the realized
covolatility matrix relies upon synchronous observations and can not be readily extended for asynchronous data. Therefore in order to evaluate the realized covolatility, we have to first ``synchronize'' the data. Fixed clock time and refresh time \cite{barndorff2011multivariate} samplings are two such synchronizing algorithms widely used in practice. But the realized covariance, evaluated on a synchronous grid, is biased \cite{hayashi2005covariance}. \cite{hayashi2005covariance} proposed an unbiased estimator of the Integrated covolatility that is applicable on intraday data without a need for synchronization. We will call this estimator as the {\it Hayashi-Yoshida estimator}. Although in presence of microstructure noise Hayashi-Yoshida estimator is also biased, a bias-corrected version was developed \cite{voev2006integrated}.

Hayashi-Yoshida estimator has good asymptotic properties as long as the data comes from an underlying low-dimensional diffusion process. But as the dimension of the data increases the estimator becomes inefficient. Developing a good estimator of high dimensional covolatility is challenging unless we impose some structure. A consistent and positive definite estimator is proposed based on blocking and regularization techniques \cite{hautsch2012blocking}. The central idea is to obtain one large covariance matrix from a series of smaller covariance matrices, each based on different sampling frequency. Shrinkage estimator of the covariance matrix with optimal shrinkage intensity, which is also important for portfolio optimization, also reduces the estimation error significantly \cite{ledoit2003honey}. Many other modified shrinkage estimator with good asymptotic properties are proposed and applied in financial context \cite{ledoit2004well, ledoit2020analytical}. The mixed frequency factor models, which uses high-frequency data to estimate factor covariance and low-frequency data to estimate the factor loadings, are also used to estimate high dimensional covariance matrices \cite{bannouh2012realized}. The composite realized Kernel approach that estimates each entry of ICV matrix optimally (in terms of bandwidth and data loss) has been proposed and asymptotic properties are established \cite{lunde2011econometric}. High-dimensionality affects the subsequent calculations of many important quantities based on covariance matrix. \cite{el2010high} showed that high-dimensionality affects the solution of Markowitz problem and results in underestimation of risk.

Instead of imposing a structure, an alternative avenue of investigating a high dimensional covariance matrix is to study its spectral distribution.
\cite{zheng2011estimation} established the limiting spectral distribution of realized covariance matrix obtained from synchronized data. Recently an asymptotic relationship has been established between the limiting spectral distributions of the true sample covariance matrix and noisy sample covariance matrix \cite{xia2018inference}. \cite{wang2021estimation} studied the estimation of integrated covariance matrix based on noisy high-frequency data with multiple transactions using random matrix theory. \cite{robert2010limiting} obtained the limiting spectral distribution of the covariance matrices of time-lagged processes. The limiting spectral distribution of sample covariance matrix was also derived under VARMA(p,q) model assumption \cite{wang2011limiting}. In this paper, we establish the limiting spectral distribution for the Hayashi-Yoshida estimator which has not yet been studied. Rest of the paper is organized as follows. In section \ref{sec:ICV_async}, we discuss the background of the problem. Section \ref{sec:spectral_distn} deals with a very brief introduction to random matrix theory. In section \ref{sec:LSD_main_theorem}, we determine the limiting spectral distribution of high-dimensional Hayashi-Yoshida estimator. Simulated data analysis results are presented in section \ref{sec:Simulation}. The summary of this work and a brief discussion on some further directions are given in section \ref{sec:conclusion}.

\section{Integrated Covariance Matrix and Asynchronicity}\label{sec:ICV_async}

Suppose, we have $p$ stocks, whose price
processes are denoted by $S_{t}^{j}$ for $j=1,...,p$. and define the $j$th log price process as $X_{t}^{j}:=\mathrm{log}S_{t}^{j}$. Let $X_{t}=(X_{t}^{1},...,X_{t}^{p})^{T}$.
Then we can model $X_{t}$ as a $p$-dimensional diffusion process
described as

\begin{equation}
dX_{t}=\mu_{t}dt+\sigma dW_{t}\label{eq:diffusion_process}
\end{equation}
where $\mu_{t}$ is a $p$ dimensional drift process and $\sigma$ is a $p$x$p$ matrix, called instantaneous covolatility process. $W_{t}$ is a $p$ dimensional standard Brownian motion. The Integrated covariance (ICV) matrix, our parameter of interest, is defined by
\begin{equation}
\Sigma_{p}=\int_{0}^{1}\sigma\sigma^{T}dt.
\end{equation}
In univariate case, the most widely used estimator of integrated variance is called the \emph{Realized variance}. For $p$ stocks, analogous covariance estimator can be defined in the following way.

\subsection{Realized covariance}
Note that the transactions in each stock occur  at random time points. Let $n_{i}$ be the number of observations for the $i$th stock.
The arrival time of the $l$th observation of the $i$th stock is denoted by $t_{l}^{i}$ . When the observations are assumed to be synchronous i.e. $t_{l}^{i}=t_{l}$ for $\forall i$,  the Realized Covariance (RCV) matrix can be defined as the following:

\begin{equation}
\begin{split}
\Sigma_{p}^{RCV} & = \sum_{l=1}^{n}\Delta X_{l}\Delta X_{l}^{T}, \text{where} \\
\Delta X_{l} & =\begin{pmatrix}\begin{array}{c}
\Delta X_{l}^{1}\\
\Delta X_{l}^{2}\\
.\\
.\\
\Delta X_{l}^{p}
\end{array}\end{pmatrix}  =\begin{pmatrix}\begin{array}{c}
X_{t_{l}}^{1}-X_{t_{l-1}}^{1}\\
.\\
.\\
.\\
X_{t_{l}}^{p}-X_{t_{l-1}}^{p}
\end{array}\end{pmatrix}.
\end{split}
\end{equation}

\begin{figure}
    \centering
\begin{tabular}{p{.6\textwidth}}
\begin{tikzpicture}[scale=1,xscale=2]
\node[shape=circle,draw=black,fill=black!5] (1) at (0,0) {$t_0^{1}$};
\node[shape=circle,draw=black,fill=black!5] (2) at (0.9,0) {$t_1^{1}$};
\node[shape=circle,draw=black,fill=black!5] (3) at (2.2,0) {$t_2^{1}$};
\node[shape=circle,draw=black,fill=black!5] (4) at (2.8,0) {$t_3^{1}$};
\node[shape=circle,draw=black,fill=black!5] (5) at (3.5,0) {$t_4^{1}$};
\node[shape=circle,draw=black,fill=black!5] (6) at (4.5,0) {$t_5^{1}$};
\node[shape=rectangle,draw=black,fill=black!5] (a) at (0,-3) {$t_0^{2}$};
\node[shape=rectangle,draw=black,fill=black!5] (b) at (1.2,-3) {$t_1^{2}$};
\node[shape=rectangle,draw=black,fill=black!5] (c) at (1.8,-3) {$t_2^{2}$};
\node[shape=rectangle,draw=black,fill=black!5] (d) at (5,-3) {$t_3^{2}$};
\coordinate (7) at (5,0);
\draw[black, thick] (1) -- (2);
\draw[black, thick] (2) -- (3);
\draw[black, thick] (3) -- (4);
\draw[black, thick] (4) -- (5);
\draw[black, thick] (5) -- (6);
\draw[black, thick] (6) -- (7);
\draw[black, thick] (a) -- (b);
\draw[black, thick] (b) -- (c);
\draw[black, thick] (c) -- (d);
\draw[black, dashed] (d) -- (7);
\draw[black,dashed] (1) -- (a);
\draw[black,dashed] (2) -- (0.9,-3);
\draw[black,dashed] (1.2,0) -- (b);
\draw[black,dashed] (3) -- (2.2,-3);
\draw[black,dashed] (1.8,0) -- (c);
\draw[black,dashed] (4) -- (2.8,-3);
\draw[black,dashed] (5) -- (3.5,-3);
\draw[black,dashed] (6) -- (4.5,-3);
\end{tikzpicture}
\vspace{.05in}\\
	\centering (a) Nonsynchronous observations
\end{tabular}\\
\vspace{0.5cm}
\begin{tabular}{p{.6\textwidth}}
\begin{tikzpicture}[scale=1,xscale=2]
\node[shape=circle,draw=red,fill=red!5] (1) at (0,0) {$t_0^{1}$};
\node[shape=circle,draw=blue,fill=blue!5] (2) at (0.9,0) {$t_1^{1}$};
\node[shape=circle,draw=orange,fill=orange!5] (3) at (2.2,0) {$t_2^{1}$};
\node[shape=circle,draw=black,fill=black!5] (4) at (2.8,0) {$t_3^{1}$};
\node[shape=circle,draw=black,fill=black!5] (5) at (3.5,0) {$t_4^{1}$};
\node[shape=circle,draw=green,fill=green!5] (6) at (4.5,0) {$t_5^{1}$};
\node[shape=rectangle,draw=red,fill=red!5] (a) at (0,-3) {$t_0^{2}$};
\node[shape=rectangle,draw=blue,fill=blue!5] (b) at (1.2,-3) {$t_1^{2}$};
\node[shape=rectangle,draw=orange,fill=orange!5] (c) at (1.8,-3) {$t_2^{2}$};
\node[shape=rectangle,draw=green,fill=green!5] (d) at (5,-3) {$t_3^{2}$};
\coordinate (7) at (5,0);
\draw[black, thick] (1) -- (2);
\draw[black, thick] (2) -- (3);
\draw[black, thick] (3) -- (4);
\draw[black, thick] (4) -- (5);
\draw[black, thick] (5) -- (6);
\draw[black, thick] (6) -- (7);
\draw[black, thick] (a) -- (b);
\draw[black, thick] (b) -- (c);
\draw[black, thick] (c) -- (d);
\draw[black,thick, <->] (1) -- (a);
\draw[black,thick, <->] (2) -- (b);
\draw[black,thick, <->] (3) -- (c);
\draw[black,thick, <->] (6) -- (d);
\draw[black,dashed] (2) -- (0.9,-3);
\draw[black,dashed] (1.2,0) -- (b);
\draw[black,dashed] (3) -- (2.2,-3);
\draw[black,dashed] (1.8,0) -- (c);
\draw[black,dashed] (4) -- (2.8,-3);
\draw[black,dashed] (5) -- (3.5,-3);
\draw[black,dashed] (5,0) -- (d);
\draw[black,dashed] (6) -- (4.5,-3);
\end{tikzpicture}
\vspace{.05in}\\
	\centering (b) Synchronized observations illustrated by arrows.
\end{tabular}
\caption{(a) Illustration of asynchronous arrival or transaction times for two stocks are shown. The opens at $t^1_0=t^2_0=0$ and the corresponding stock price at time $0$ can be taken as previous day's closing price. Transaction times of the first stock are shown by circles. Transaction times of the second stock are shown by rectangles. (b) The synchronized pairs (having the same colour code) are indicated by arrows.}
    \label{fig:asynchronicity_and_sampling_times}
\end{figure}

\subsection{Hayashi-Yoshida estimator}
For asynchronous intraday data, the Realized covariance can not be directly calculated unless we synchronize the data by some ad hoc method. This means that we have to throw away some of the observations such that synchronized vectors of observations can be formed. In Fig. \ref{fig:asynchronicity_and_sampling_times}, we illustrate this for the bivariate case. Fig. \ref{fig:asynchronicity_and_sampling_times}(a) shows how nonsynchronous data would look like. The circles and quadrangles represent the transaction times for the first and second stock respectively. Synchronization of transaction times, as indicated by the arrows, are shown in Fig. \ref{fig:asynchronicity_and_sampling_times}(b). A synchronized dataset can be formed by pairing the stock prices corresponding to the synchronized time points. For example, two consecutive observations can be $(X^1_{t^1_2},X^2_{t^2_2})$ and $(X^1_{t^1_5},X^2_{t^2_3})$. This is equivalent to ``pretending'' that $X^1_{t^1_5}$ is observed at $t^2_3$ instead of at $t^1_5$ and similarly, $X^2_{t^2_2}$ is observed at $t^1_2$ instead of at $t^2_2$.\footnote{For this reason, synchronization methods can be expressed as a problem of choosing a set of sampling times $\{\tau_1,\tau_2,...,\tau_n\}$, from the set $\mathcal{T} = \{t^1_1,...,t^1_{n_1}\}\cup\{t^2_1,...,t^2_{n_2}\}$. The corresponding price of each stock at $\tau_i$ is taken as the price observed previous to $\tau_i$.} We can see from Fig \ref{fig:asynchronicity_and_sampling_times}(b) that for the first stock, two observations at time $t^1_3$ and $t^1_4$ are not synchronized with any time point in the second stock and therefore can be excluded from the study. Extending this to $p$ stocks, let us denote that the number of resulting synchronized vectors is $n$. Then it is evident that $n\leq \{n_1,n_2,...,n_p\}$.

In Fig. \ref{fig:asynchronicity_and_sampling_times}(b), we have applied a particular synchronization method called \emph{refresh time} sampling \cite{ait2010high, fan2012vast}.

\cite{hayashi2005covariance} proposed an alternative estimator ($\Sigma_{p}^{HY}$) of ICV matrix that does not require the dataset to be synchronized and therefore can be directly applied on the asynchronous data. Before defining it for high-dimension, we introduce it for the bivariate case. In the following expression, instead of writing $X^1$ and $X^2$ we simply write $X$ and $Y$. Now, for two stocks the Hayashi-Yoshida estimator is defined as the following way:
\begin{equation}
    \Sigma_{2}^{HY} = \sum_{k,l}\Delta X_{k}\Delta Y_{l}\times {\bf 1}\{(t_{k-1}^{i},t_{k}^{i})\cap (t_{l-1}^{i},t_{l}^{i})\neq \phi \},\label{eq:HY_for_two_stocks}
\end{equation}
where ${\bf 1}\{(t_{k-1}^{i},t_{k}^{i})\cap (t_{l-1}^{i},t_{l}^{i})\neq \phi \}$ is an indicator function that takes value 1 when the condition is satisfied. Fig. \ref{fig:Hayashi_Yoshida_refresh_time} illustrates the computation. $\Delta x_1\times \Delta y_1$ will contribute to the sum in Eq.~\eqref{eq:HY_for_two_stocks} as $(t^2_0,t^2_1)$ and $(t^1_0,t^1_1)$ are overlapping intervals. But $\Delta x_1\times \Delta y_2$ will not contribute to the sum in Eq.~\eqref{eq:HY_for_two_stocks} as the intervals $(t^2_1,t^2_2)$ and $(t^1_0,t^1_1)$ are non-overlapping.
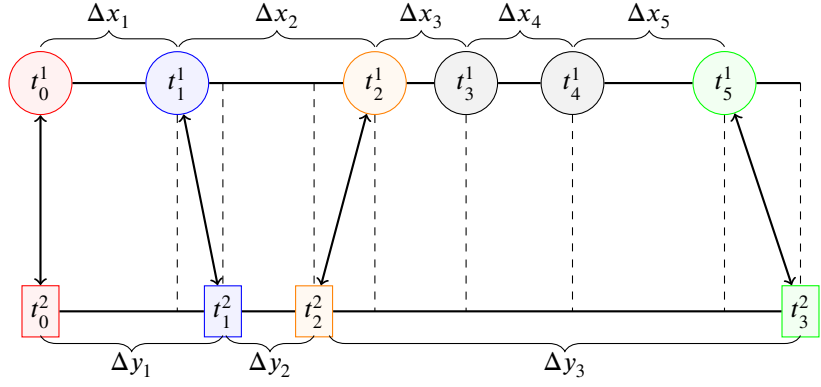
\begin{figure}
\begin{center}
\begin{tikzpicture}[scale=1,xscale=2]
\node[shape=circle,draw=red,fill=red!5] (1) at (0,0) {$t_0^{1}$};
\node[shape=circle,draw=blue,fill=blue!5] (2) at (0.9,0) {$t_1^{1}$};
\node[shape=circle,draw=orange,fill=orange!5] (3) at (2.2,0) {$t_2^{1}$};
\node[shape=circle,draw=black,fill=black!5] (4) at (2.8,0) {$t_3^{1}$};
\node[shape=circle,draw=black,fill=black!5] (5) at (3.5,0) {$t_4^{1}$};
\node[shape=circle,draw=green,fill=green!5] (6) at (4.5,0) {$t_5^{1}$};
\node[shape=rectangle,draw=red,fill=red!5] (a) at (0,-3) {$t_0^{2}$};
\node[shape=rectangle,draw=blue,fill=blue!5] (b) at (1.2,-3) {$t_1^{2}$};
\node[shape=rectangle,draw=orange,fill=orange!5] (c) at (1.8,-3) {$t_2^{2}$};
\node[shape=rectangle,draw=green,fill=green!5] (d) at (5,-3) {$t_3^{2}$};
\coordinate (7) at (5,0);
\draw[black, thick] (1) -- (2);
\draw[black, thick] (2) -- (3);
\draw[black, thick] (3) -- (4);
\draw[black, thick] (4) -- (5);
\draw[black, thick] (5) -- (6);
\draw[black, thick] (6) -- (7);
\draw[black, thick] (a) -- (b);
\draw[black, thick] (b) -- (c);
\draw[black, thick] (c) -- (d);
\draw[black,thick, <->] (1) -- (a);
\draw[black,thick, <->] (2) -- (b);
\draw[black,thick, <->] (3) -- (c);
\draw[black,thick, <->] (6) -- (d);
\draw[black,dashed] (2) -- (0.9,-3);
\draw[black,dashed] (1.2,0) -- (b);
\draw[black,dashed] (3) -- (2.2,-3);
\draw[black,dashed] (1.8,0) -- (c);
\draw[black,dashed] (4) -- (2.8,-3);
\draw[black,dashed] (5) -- (3.5,-3);
\draw[black,dashed] (5,0) -- (d);
\draw[black,dashed] (6) -- (4.5,-3);
\draw [decorate,decoration={brace,amplitude=6pt,mirror,raise=2ex}]
(0,-3) -- (1.2,-3)  node[midway,yshift=-2em]{$\Delta y_1$};
\draw[decorate,decoration={brace,amplitude=6pt,mirror,raise=2ex}]
(1.22,-3) -- (1.8,-3) node[midway,yshift=-2em]{$\Delta y_2$};
\draw[decorate,decoration={brace,amplitude=6pt,mirror,raise=2ex}]
(1.9,-3) -- (5,-3) node[midway,yshift=-2em]{$\Delta y_3$};
\draw [decorate,decoration={ brace,amplitude=6pt,raise=3ex}] (0,0)--(0.9,0) node[midway,yshift=2.5em]{$\Delta x_1$};
\draw [decorate,decoration={ brace,amplitude=6pt,raise=3ex}] (0.91,0)--(2.2,0) node[midway,yshift=2.5em]{$\Delta x_2$};
\draw [decorate,decoration={ brace,amplitude=6pt,raise=3ex}] (2.21,0)--(2.8,0) node[midway,yshift=2.5em]{$\Delta x_3$};
\draw [decorate,decoration={ brace,amplitude=6pt,raise=3ex}] (2.81,0)--(3.5,0) node[midway,yshift=2.5em]{$\Delta x_4$};
\draw [decorate,decoration={ brace,amplitude=6pt,raise=3ex}] (3.51,0)--(4.5,0) node[midway,yshift=2.5em]{$\Delta x_5$};
\end{tikzpicture}
\end{center}
\caption{Illustration of asynchronous arrival or transaction times for two stocks are shown. The opens at $t^1_0=t^2_0=0$ and the corresponding stock price at time $0$ can be taken as previous day's closing price. Transaction times of the first stock are shown by circles. Transaction times of the second stock are shown by rectangles. The returns on each interarrival (shown by braces) are denoted by $\Delta x_k = x_k-x_{k-1}$ and $\Delta y_l = y_l-y_{l-1}$. When the intervals $(t^1_{k-1},t^1_{k})$ and $(t^2_{l-1},t^2_{l})$ have an nonempty intersection, $\Delta x_k \times \Delta y_l$ will contribute to the Hayashi-Yoshida covariance.}
    \label{fig:Hayashi_Yoshida_refresh_time}
\end{figure}

\subsubsection{Hayashi-Yoshida covariance and Refresh-time sampling}\label{sec:HY_n_refresh_time}
Even though the Hayashi-Yoshida estimator does not require prior synchronization of intraday data, we will show that the estimator still throws away some of the data points. Moreover these data-points are exactly the same as thrown by refresh-time sampling. To see this let us consider the case as shown in Fig. \ref{fig:Hayashi_Yoshida_refresh_time}. The Hayashi-Yoshida estimate is:
\begin{equation}
    \begin{split}
    \sigma^{HY}(X,Y) & = \Delta x_1 \Delta y_1 + \Delta x_2 \Delta y_1 + \Delta x_2 \Delta y_2 + \Delta x_2 \Delta y_3 + \Delta x_3 \Delta y_3 + \Delta x_4 \Delta y_3 + \Delta x_5 \Delta y_3\\
    & = \Delta x_1 \Delta y_1 + \Delta x_2 \Delta y_1 + \Delta x_2 \Delta y_2 + \Delta x_2 \Delta y_3 + (\Delta x_3 + \Delta x_4 + \Delta x_5) \Delta y_3
    \end{split}
\end{equation}
But $\Delta x_3 + \Delta x_4 + \Delta x_5$ is just the difference between log-price at $t^1_5$ and log-price at $t^1_3$ which doesn't require any information on stock price anytime time in between. Therefore, although the Hayashi-Yoshida estimator doesn't require presynchronization, it actually throws away the exact same observation as thrown by refresh-time sampling. As a consequence, the value of Hayashi-Yoshida covariance on full data will be equal to Hayashi-Yoshida estimator on the set of refresh-time pairs. Synchronizing the data using refresh-time sampling before computing the covariance can reduce the computational cost quite significantly.

When we move away from bivariate case to higher dimension, synchronize every pair of variables separately would not be very efficient. It would be preferable to synchronize the data for all the stocks simultaneously. This can be achieved by applying ``all refresh method'' which results in a synchronous sampled time points $\{\tilde{t}_1, \tilde{t}_2,...,\tilde{t}_n\}$ defined in the following way:
\begin{equation}
    \begin{split}
        \tilde{t}_{j+1} = \underset{1\leq i\leq p}{\text{max}} t^i_{N_i(\tilde{t}_j)}, \label{eq:refresh_algorithm}
    \end{split}
\end{equation}
where $N_i(t)$ is the number of observation before time $t$ \cite{guo2017quantitative, barndorff2011subsampling}. In this paper, we will define the Hayashi-Yoshida estimator on refresh time sampling times. The theoretical implication is that- given the synchronized data, we can now assume the number of observations of each stock to be the same. We will denote this common sample size as $n$.

 For $p$ stocks the Hayashi-Yoshida estimator is defined as the following :

\begin{equation}
\begin{split}
    \Sigma_{p}^{HY} & =\sum_{k,l}\Delta X_{k}\Delta X_{l}^{T}\circ I(k,l), \text{where} \\
    \Delta X_{l} & =\begin{pmatrix}\begin{array}{c}
\Delta X_{l}^{1}\\
\Delta X_{l}^{2}\\
.\\
.\\
\Delta X_{l}^{p}
\end{array}\end{pmatrix} =\begin{pmatrix}\begin{array}{c}
X_{t_{l}^{1}}^{1}-X_{t_{l-1}^{1}}^{1}\\
.\\
.\\
.\\
X_{t_{l}^{p}}^{p}-X_{t_{l-1}^{p}}^{p}
\end{array}\end{pmatrix}
\end{split}
\end{equation}
 and `$\circ$' is the Hadamard product and $I(k,l)$ is a $p\times p$ matrix with $(i,j)^{th}$ element is the indicator function involving $k^{th}$interarrival of $i^{th}$stock and $l^{th}$interarrival of $j^{th}$ stock: $I(I_{k}^{i}\cap I_{l}^{j}\neq\phi)$, where $I_{k}^{i}=(t_{k-1}^{i},t_{k}^{i})$. In other words if two interarrivals intersect then product $\Delta X_{k}^{i}\Delta X_{l}^{j\thinspace T}$ will contribute to the sum. In Fig. \ref{fig:Hayashi_Yoshida_refresh_time},  $I_{k}^{1}=(t_{k-1}^{1},t_{k}^{1}),~k\in \{1,2,3,4,5\}$ and $I_{l}^{2}=(t_{l-1}^{2},t_{l}^{2}),~l\in \{1,2,3\}$ are shown.

\subsection{Scaled Realized Covariance estimator}\label{sec:SRCV}
In this section we show the ``closeness'' of the Hayashi-Yoshida estimator with a scaled realized estimator which is motivated from the intraday covariance estimator proposed in \cite{chakrabarti2019copula}. We determine the scaling coefficients for the bivariate case. The result will be key in our proof of the Limiting Spectral distribution. For the bivariate case, let us denote the log price for two stocks at a particular time $t$ as $(X_t,Y_t)$.  Following \cite{chakrabarti2019copula}, we synchronize the data for two stocks in the following fashion:

\begin{center}
\fbox{\begin{minipage}{0.8\textwidth}
\textbf{Algorithm} ($\mathcal{A}_{0}$):

1. For $i=1$, assign $k_{i}^{1}=1$ and $k_{i}^{2}=1$.

2. While $k^1_i \leq n_1$ and $k^2_i \leq n_2$:
\begin{itemize}
    \item If $t_{k_{i}^{2}}^{2}>t_{k_{i}^{1}}^{1}$ then find $m=\mathrm{max}\{j:\ t_{j}^{1}<t_{k_{i}^{2}}^{2}\}$.
The $i$th pair will be $(X_{t_{m}^{1}},Y_{t_{k_{i}^{2}}^{2}})$.
Modify $k_{i}^{1}=m$.
    \item If $t_{k_{i}^{2}}^{2} \leq t_{k_{i}^{1}}^{1}$ then find $m=\mathrm{max}\{j:\ t_{j}^{2}<t_{k_{i}^{1}}^{1}\}$.
The $i$th pair will be $(X_{t_{k_{i}^{1}}^{1}},Y_{t_{m}^{2}})$.
Modify $k_{i}^{2}=m$
    \item Modify $i=i+1$. $k_{i}^{1}=k_{i}^{1}+1$ and $k_{i}^{2}=k_{i}^{2}+1$.
\end{itemize}
\end{minipage} }
\end{center}
The pairs created by this algorithm are identical to the pairs created by \emph{refresh time sampling} but accommodates more information by retaining the actual transaction times. To see this, note that in Eq.\eqref{eq:refresh_algorithm}), a common set of synchronized points $\{\tilde{t}_1,...,\tilde{t}_n\}$ are defined for all stocks. For each stock the last observed stock price prior to $\tau_i$ was taken to be the price at $\tau_i$. Therefore, from the refresh time sampling it is not possible to retrieve the actual transaction times of the synchronized pairs. Algorithm $\mathcal{A}_0$, on the other hand, keep these information. Instead of writing $(X_{t_{k_{i}^{1}}^{1}},Y_{t_{k_{i}^{2}}^{2}})$ we shall henceforth write $(X_{t(k_{i}^{1})},Y_{t(k_{i}^{2})})$.

\subsubsection{Overlapping and non-overlapping regions for return construction}
For two such consecutive synchronized pairs of stock-prices, we can now consider the bivariate return as: $\{(X_{t(k_{i}^{1})}-X_{t(k_{i-1}^{1})}), (Y_{t(k_{i}^{2})}-Y_{t(k_{i-1}^{2})}): i=1,2,...,n \}$. Note that, in this bivariate return vector, the first component (for $X$) is defined on the interval $\big(t(k_{i-1}^{1}), t(k_{i}^{1})\big)$ and the return on the $Y$ is defined on the interval $\big(t(k_{i-1}^{2}), t(k_{i}^{2})\big)$. It can be shown the \emph{overlap} and \emph{nonoverlapping} parts of these two intervals play a crucial role in bias of the estimated covariance \cite{chakrabarti2019copula}. To define the \emph{overlap} we first illustrate four possible configurations of the intervals. In Fig. \ref{fig:config_n_overlap} we show four such configurations of intervals corresponding to a particular return vector, constructed from synchronized pairs of observations. More formally, suppose $X_{t(k_{i}^{1})}-X_{t(k_{i-1}^{1})}=\sum_{i=m}^{l}(X_{t_{i+1}}-X_{t_{i}})$ for some $m$ and $l$. Then one of these four configurations is true:

\begin{equation}\left[
\begin{array}{rcl}1. \quad  Y_{t(k_{i}^{2})}-Y_{t(k_{i-1}^{2})}&=&\sum_{i=m+1}^{l-1}(Y_{t_{i+1}}-Y_{t_{i}})\\
2. \quad Y_{t(k_{i}^{2})}-Y_{t(k_{i-1}^{2})}&=&\sum_{i=m-1}^{l-1}(Y_{t_{i+1}}-Y_{t_{i}})\\
3. \quad Y_{t(k_{i}^{2})}-Y_{t(k_{i-1}^{2})}&=&\sum_{i=m+1}^{l+1}(Y_{t_{i+1}}-Y_{t_{i}})\\
4. \quad  Y_{t(k_{i}^{2})}-Y_{t(k_{i-1}^{2})}&=&\sum_{i=m-1}^{l+1}(Y_{t_{i+1}}-Y_{t_{i}})\end{array}\right]\label{eqn:config}
\end{equation}

Given this set of possible configurations, we define a random variable $L_i$, denoting the overlapping
time interval of $i$th interarrivals corresponding to $X_{t(k_{i}^{1})}-X_{t(k_{i-1}^{2})}$ and $Y_{t(k_{i}^{2})}-Y_{t(k_{i-1}^{2})}$ as
\begin{equation}
L_i=\begin{cases}
t(k_{i}^{2})-t(k_{i-1}^{2}) & \mathrm{if}\ Y_{t(k_{i}^{2})}-Y_{t(k_{i-1}^{1})}=\sum_{i=m+1}^{l-1}(Y_{t_{i+1}}-Y_{t_{i}})\\
t(k_{i}^{2})-t(k_{i-1}^{1}) & \mathrm{if}\ Y_{t(k_{i}^{2})}-Y_{t(k_{i-1}^{1})}=\sum_{i=m-1}^{l-1}(Y_{t_{i+1}}-Y_{t_{i}})\\
t(k_{i}^{1})-t(k_{i-1}^{2}) & \mathrm{if}\ Y_{t(k_{i}^{2})}-Y_{t(k_{i-1}^{1})}=\sum_{i=m+1}^{l+1}(Y_{t_{i+1}}-Y_{t_{i}})\\
t(k_{i}^{1})-t(k_{i-1}^{1}) & \mathrm{if}\ Y_{t(k_{i}^{2})}-Y_{t(k_{i-1}^{1})}=\sum_{i=m-1}^{l+1}(Y_{t_{i+1}}-Y_{t_{i}})
\end{cases}
\end{equation}
Fig. \ref{fig:config_n_overlap} illustrates the overlapping regions for all four configurations described in Eq. \ref{eqn:config}.

\begin{figure}
\centering
\begin{tabular}{p{.4\textwidth}p{.49\textwidth}}
		\parbox{.4\textwidth}{
			\centering
			\begin{tikzpicture}
\node[shape=circle, inner sep=0pt,minimum size=1cm,draw=red,fill=red!5] (1) at (-0.3,0) {${\scriptstyle t(k_{i-1}^1)}$};
\node[shape=circle, inner sep=0pt,minimum size=1cm,draw=green,fill=green!5] (2) at (3.4,0) {${\scriptstyle t(k_{i}^{1})}$};
\node[shape=rectangle, inner sep=0pt,minimum size=1cm,draw=red,fill=red!5] (b) at (0.8,-3) {${\scriptstyle t(k_{i-1}^2)}$};
\node[shape=rectangle, inner sep=0pt,minimum size=1cm,draw=green,fill=green!5] (c) at (2.6,-3) {${\scriptstyle t(k_{i}^{2})}$};
\coordinate (0) at (-1.1,0);
\coordinate (a) at (-1,-3);
\coordinate (3) at (4.2,0);
\coordinate (d) at (4.2,-3);
\draw[black,thick] (0) -- (1);
\draw[black, thick] (1) -- (2);
\draw[black,thick] (a) -- (b);
\draw[black, thick] (b) -- (c);
\draw[black, thick] (2) -- (3);
\draw[black, thick] (c) -- (d);
\draw[black,thick, <->] (1) -- (b);
\draw[black,thick, <->] (2) -- (c);
\draw[black,dashed] (1) -- (-0.3,-3);
\draw[black,dashed] (2) -- (3.4,-3);
\draw[black,dashed] (b) -- (0.8,0);
\draw[black,dashed] (c) -- (2.6,0);
\draw[blue,ultra thick, <->] (0.8,-1.5) -- (2.6,-1.5)  node[midway,yshift=0.8em]{$\textcolor{black}{L_i}$};
\end{tikzpicture}
} &
\parbox{.4\textwidth}{
			\centering
			\begin{tikzpicture}
\node[shape=circle, inner sep=0pt,minimum size=1cm,draw=red,fill=red!5] (1) at (0,0) {${\scriptstyle t(k_{i-1}^1)}$};
\node[shape=circle, inner sep=0pt,minimum size=1cm,draw=green,fill=green!5] (2) at (2.4,0) {${\scriptstyle t(k_{i}^{1})}$};
\node[shape=rectangle, inner sep=0pt,minimum size=1cm,draw=red,fill=red!5] (b) at (1.2,-3) {${\scriptstyle t(k_{i-1}^2)}$};
\node[shape=rectangle, inner sep=0pt,minimum size=1cm,draw=green,fill=green!5] (c) at (3.5,-3) {${\scriptstyle t(k_{i}^{2})}$};
\coordinate (0) at (-1,0);
\coordinate (a) at (-1,-3);
\coordinate (3) at (4.2,0);
\coordinate (d) at (4.2,-3);
\draw[black,thick] (0) -- (1);
\draw[black, thick] (1) -- (2);
\draw[black,thick] (a) -- (b);
\draw[black, thick] (b) -- (c);
\draw[black, thick] (2) -- (3);
\draw[black, thick] (c) -- (d);
\draw[black,thick, <->] (1) -- (b);
\draw[black,thick, <->] (2) -- (c);
\draw[black,dashed] (1) -- (0,-3);
\draw[black,dashed] (2) -- (2.4,-3);
\draw[black,dashed] (b) -- (1.2,0);
\draw[black,dashed] (c) -- (3.5,0);
\draw[blue,ultra thick, <->] (1.2,-1.5) -- (2.4,-1.5)  node[midway,yshift=0.8em]{$\textcolor{black}{L_i}$};
\end{tikzpicture}
} \vspace{.1in}\\
	\centering (a) &\centering (b)
\end{tabular} \\
\begin{tabular}{p{.4\textwidth}p{.49\textwidth}}
\parbox{.4\textwidth}{
			\centering
			\begin{tikzpicture}
\node[shape=circle, inner sep=0pt,minimum size=1cm,draw=red,fill=red!5] (1) at (1,0) {${\scriptstyle t(k_{i-1}^1)}$};
\node[shape=circle, inner sep=0pt,minimum size=1cm,draw=green,fill=green!5] (2) at (2.5,0) {${\scriptstyle t(k_{i}^{1})}$};
\node[shape=rectangle, inner sep=0pt,minimum size=1cm,draw=red,fill=red!5] (b) at (0,-3) {${\scriptstyle t(k_{i-1}^2)}$};
\node[shape=rectangle, inner sep=0pt,minimum size=1cm,draw=green,fill=green!5] (c) at (3.6,-3) {${\scriptstyle t(k_{i}^{2})}$};
\coordinate (0) at (-1.1,0);
\coordinate (a) at (-1,-3);
\coordinate (3) at (4.2,0);
\coordinate (d) at (4.2,-3);
\draw[black,thick] (0) -- (1);
\draw[black, thick] (1) -- (2);
\draw[black,thick] (a) -- (b);
\draw[black, thick] (b) -- (c);
\draw[black, thick] (2) -- (3);
\draw[black, thick] (c) -- (d);
\draw[black,thick, <->] (1) -- (b);
\draw[black,thick, <->] (2) -- (c);
\draw[black,dashed] (1) -- (1,-3);
\draw[black,dashed] (2) -- (2.5,-3);
\draw[black,dashed] (b) -- (0,0);
\draw[black,dashed] (c) -- (3.6,0);
\draw[blue,ultra thick, <->] (1,-1.5) -- (2.5,-1.5)  node[midway,yshift=0.8em]{$\textcolor{black}{L_i}$};
\end{tikzpicture}
} &
\parbox{.4\textwidth}{
			\centering
			\begin{tikzpicture}
\node[shape=circle, inner sep=0pt,minimum size=1cm,draw=red,fill=red!5] (1) at (1,0) {${\scriptstyle t(k_{i-1}^1)}$};
\node[shape=circle, inner sep=0pt,minimum size=1cm,draw=green,fill=green!5] (2) at (3.8,0) {${\scriptstyle t(k_{i}^{1})}$};
\node[shape=rectangle, inner sep=0pt,minimum size=1cm,draw=red,fill=red!5] (b) at (0,-3) {${\scriptstyle t(k_{i-1}^2)}$};
\node[shape=rectangle, inner sep=0pt,minimum size=1cm,draw=green,fill=green!5] (c) at (2.4,-3) {${\scriptstyle t(k_{i}^{2})}$};
\coordinate (0) at (-1,0);
\coordinate (a) at (-1,-3);
\coordinate (3) at (4.2,0);
\coordinate (d) at (4.2,-3);
\draw[black,thick] (0) -- (1);
\draw[black, thick] (1) -- (2);
\draw[black,thick] (a) -- (b);
\draw[black, thick] (b) -- (c);
\draw[black, thick] (2) -- (3);
\draw[black, thick] (c) -- (d);
\draw[black,thick, <->] (1) -- (b);
\draw[black,thick, <->] (2) -- (c);
\draw[black,dashed] (1) -- (1,-3);
\draw[black,dashed] (2) -- (3.8,-3);
\draw[black,dashed] (b) -- (0,0);
\draw[black,dashed] (c) -- (2.4,0);
\draw[blue,ultra thick, <->] (1,-1.5) -- (2.4,-1.5)  node[midway,yshift=0.8em]{$\textcolor{black}{L_i}$};
\end{tikzpicture}
}
\vspace{.1in}\\
	 \centering (c)&\centering (d)
	\end{tabular}
	\caption{Four configurations described in Eq.\eqref{eqn:config}. For each case the overlapping interval is shown as $L_i$.}\label{fig:config_n_overlap}
\end{figure}
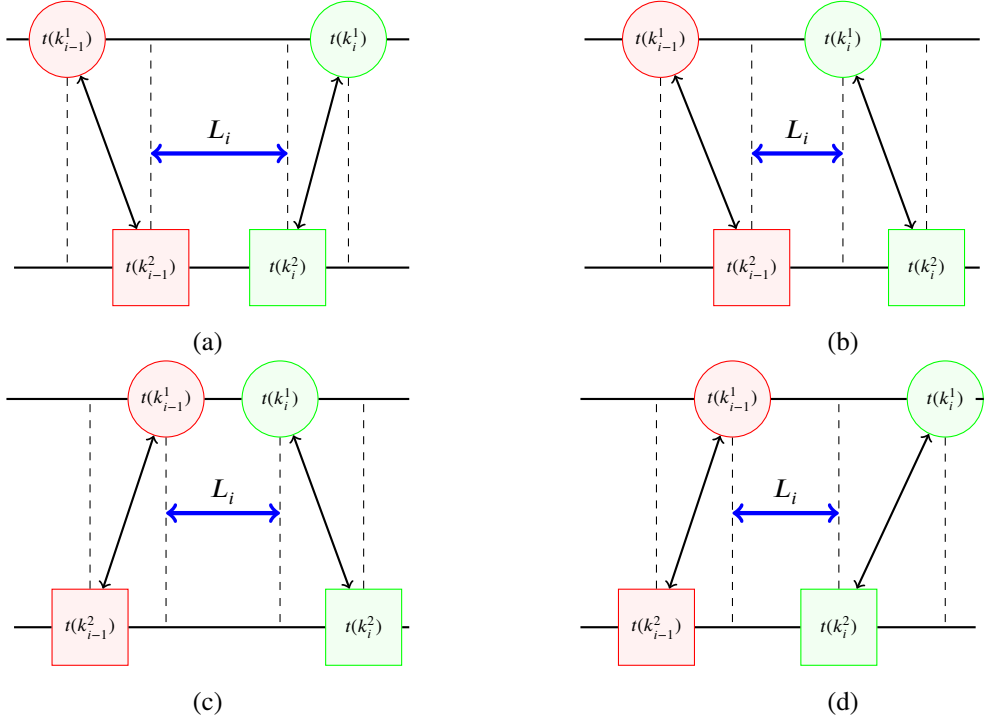

The next theorem says that in presence of high-frequency data, the performance of Hayashi-Yoshida covariance would be very similar to a scaled realized covariance. We make the following set of assumptions:
\begin{enumerate}
    \item[($\mathcal{C}_1$):] The log return process follows independent and stationary increment property.
    \item[($\mathcal{C}_2$):] The observation times (arrival process) of two stocks are independent Renewal processes and $n\rightarrow\infty$ as $n_{1},n_{2}\rightarrow$ $\infty$.
    \item[($\mathcal{C}_3$):] Estimation is based on paired data obtained by algorithm $A_0$.
\end{enumerate}

Define  a scaled (pairwise) realized covariance (SRCV) on a synchronized data $\big\{\big(X_{t(k^1_i)}, Y_{t(k^2_i)}\big):i=1,2,..,n\big\}$ as follows:
\begin{equation}
    \text{SRCV}(X,Y) = \sum_i \psi_i \Delta X_{t(k^1_i)} \Delta Y_{t(k^2_i)},\label{eq:SRCV}
\end{equation}
where
\begin{equation*}
    \psi_i = \frac{\sqrt{(t(k^1_i)-t(k^1_{i-1}))(t(k^1_i)-t(k^1_{i-1}))}}{L_i}.
\end{equation*}
The following theorem says that for all practical purposes this estimator performs as good as the Hayashi-Yoshida estimator.
\begin{theorem}
Under the assumptions $\mathcal{C}_{1}-\mathcal{C}_{3},$ SRCV is a consistent estimator of the pairwise integrated covariance. \label{thm:scaled_realized_HY}
\end{theorem}
Proof of Theorem \ref{thm:scaled_realized_HY} will be along the same line as in Theorem 1 of \cite{chakrabarti2019copula}.

The $\psi_i$'s in Eq.\eqref{eq:SRCV} are the scaling coefficients which are functions of the arrival (transaction) times only (does not depend on the stock prices at those time points). The conventional refresh-time periods $\tilde{t}_j$s in Eq. \ref{eq:refresh_algorithm} do not allow us to calculate these scaling coefficients. The algorithm $\mathcal{A}_0$, on the other hand, enable us to calculate the scaling coefficients as it preserves the actual arrival times of the synchronized pairs. The importance of this estimator will be evident in

As both Hayashi-Yoshida (pairwise) covariance and SRCV are consistent estimators, in presence of sufficient amount of data, they would be ``close'' to each other. For $p$-dimensional process, the SRVC matrix can be written as the following:
\begin{equation}
    SRCV_p = \sum_{i=1}^n (\Delta X_i \Delta X_i^T \circ \Psi_i), \label{eq:SRCV_p_dimension}
\end{equation}
where $\psi_i$ is a $p\times p$ symmetric matrix consisting of the pairwise scaling coefficients. This matrix will help us to make an important assumption necessary to determine the LSD of Hayashi-Yoshida matrix.

\subsection{Inconsistency in high-dimension}
From multivariate statistical theory we know that the sample covariance matrix is consistent for the population covariance matrix. In high-dimensional scenario, when the dimension grows at the same or a higher rate as the number of observations, this properties do not hold anymore. It can be shown that under a high-dimensional setup the following is true:
\begin{equation*}
    \|S-\Sigma\|\nrightarrow 0,
\end{equation*}
where $\|.\|$ is the operator norm, $S$ is the usual sample covariance matrix and $\Sigma$ is the population covariance matrix.

But what impact would this inconsistency of sample covariance matrix make on the eigenvalues and the eigenvectors? Weyl's theorem and Davis-Kahan theorems show that when the sample covariance matrix is not consistent, neither the sample eigenvalues and the eigenvectors are going to converge to their true counterpart \cite{pourahmadi2013high, yu2015useful}. Therefore it is worthwhile to study the limiting distribution of sample eigenvalues and its relation to the distribution of true eigenvalues.

For low dimensions, the Hayashi's estimator is consistent estimator of ICV matrix. But for high dimensional stock price process, when the dimension grows at the same or a higher rate as the number of observations, neither RCV nor Hayashi-Yoshida estimator is consistent anymore. It is evident from the above discussion that the sample spectrum of Hayashi-Yoshida estimator deviates significantly form the true spectrum. In this paper, we study the asymptotic behavior of the distribution of eigenvalues of the Hayashi-Yoshida matrix.

\section{Spectral Distribution}\label{sec:spectral_distn}
The \textit{empirical spectral distribution (ESD)} of a symmetric (more generally Hermitian) matrix $A_{p\times p}$ is defined as
\begin{equation*}
F_{p}^{A}(x)=\frac{1}{p}\#\{j\leq p:\ \lambda_{j}\leq x\}
\end{equation*}
where $\lambda_{j}$'s are the eigenvalues of the matrix $A$ and $\#E$ denotes the cardinality of set $E$. The limit distribution ($F$) of \emph{ESD} is called the \textit{limiting spectral distribution (LSD)}. One commonly used method of finding LSD is through Stieltjes transform.

\textbf{Stieltjes transform:} Let  $A_{p\times p}$  is a Hermitian matrix and $F_{p}^{A}$ be its ESD. Then the Stieltjes transform of $F_{p}^{A}$ is defined as
\begin{align*}
s_{p}(z) & =\int\frac{1}{x-z}dF_{p}^{A}(x)\\
 & =\frac{1}{p}tr(A-zI)^{-1}
\end{align*}
where $z\in D=\{z\in\mathbb{C}|\mathbb{I}(z)>0\}$, $\mathbb{I}(z)$ being the imaginary part of $z$. The importance of Stieltjes transformation in Random Matrix Theory is due to Theorem B.9 and Theorem B.10  \cite{bai2010spectral}).

These theorems suggest that in order to determine the {\it LSD}, it is enough to obtain the limit of the Stieltjes transform.
\section{Spectral Analysis of High Dimensional Hayashi's Estimator}\label{sec:LSD_main_theorem}
Based on our model (Eq. \ref{eq:diffusion_process}), the distribution of $\Delta X_i$ can be written like the following:
\begin{equation*}
    \Delta X_{i}=\sigma\left(\begin{array}{c}
\int_{t_{l-1}^{1}}^{t_{l}^{1}}dW_{1}\\
\int_{t_{l-1}^{2}}^{t_{l}^{2}}dW_{2}\\
.\\
.\\
\int_{t_{l-1}^{p}}^{t_{l}^{p}}dW_{p}
\end{array}\right)\stackrel{d}{=}\sigma D_i Y_{i},
\end{equation*}
where $D_i = \text{diag}\Big(\sqrt{t(k^1_i)-t(k^1_{i-1})},\sqrt{t(k^2_i)-t(k^2_{i-1})},..,\sqrt{t(k^1_p)-t(k^p_{i-1})}\Big)$ is a $p\times p$ diagonal matrix and $Y_i$ is a $p$-dimensional vector with its components being {\it iid} standard normal distribution. As a consequence, the Hayashi-Yoshida estimator has the same distribution as that of $S_0$:
\begin{equation*}
    S_0 = \sum_k \sum_l \sigma D_k Y_k Y_l^T D_l \sigma \circ I(k,l).
\end{equation*}
Hence, to determine the LSD of Hayashi-Yoshida estimator, it is enough to find the limit of the Stieltjes transform of $S_0$.

The set of assumptions ($\mathcal{A}$) necessary for determining the limiting spectral distribution of Hayashi-Yoshida estimator are given below:
\begin{enumerate}
    \item[($\mathcal{A}_1$)] $c_{n}=\frac{p}{n}\rightarrow c>0$ as $p\rightarrow\infty$.
    \item[($\mathcal{A}_2$)] $F^{\Sigma}\stackrel{d}{\Rightarrow}H$ ($\neq\mathrm{delta}$ $\mathrm{measure\ at\ 0}$)
a.s. and $H$ has a finite second order moment.
    \item[($\mathcal{A}_3$)] $Y_{l}^{j}$'s ($j=1(1)p$) are iid with mean $0$, variance $\text{1}$ and finite moments of all orders.
    \item[($\mathcal{A}_4$)] $\exists$ $\tau_{l}\in\mathbb{R^{+}}$ and let $S=\sum_{l}\tau_{l}\Sigma^{\frac{1}{2}}Y_{l}Y_{l}^{T}\Sigma^{\frac{1}{2}}$
such that
\begin{equation}
    \text{tr}(S_{0}-S)^{2}  = o(p)~ \text{a.s.} \label{eq:closeness_condition}
\end{equation}
and \begin{equation}
    \text{tr}([(S_{0}-zI)^{-1}(S-zI)^{-1}]^{2})  = O(p)~\text{a.s.} \label{eq:closeness_condition2}
\end{equation}
 where $z=i.v$ with $v$ is a sufficiently large positive number $(S_{0}-zI)^{-1}(S-zI)^{-1}$ being positive semidefinite.
    \item[($\mathcal{A}_5$)] There exists $\kappa,$ such that $\mathrm{max_{n}max_{l}}(n\tau_{l})\leq\kappa.$
Also there exists a nonnegative cadlag process $\tau'_{t}$ such that
\begin{equation}
\underset{n \rightarrow \infty}{\text{lim}}\sum_{l =1}^n \int_{\frac{l-1}{n}}^{\frac{l}{n}}|n\tau_{l}-\tau'_{t}|dt=0.
\end{equation}
    \item[($\mathcal{A}_6$)] There exists a $K<\infty$ and $\delta<1/6$ such that for all $p$, $\|\Sigma\|\leq Kp^{\delta}$ almost surely.
\end{enumerate}
 Before stating the main theorem, we present a brief motivation of the assumption $\mathcal{A}_4$. In Sec.~\ref{sec:SRCV}, we have seen that under a low dimensional setup, SRCV is also a consistent estimator for ICV. Therefore with high frequency data we can expect $S_0$ and $SRCV$ to be ``close'' to each other. The matrix $S$ replaces the matrix $\Psi_i$ in Eq. \eqref{eq:SRCV_p_dimension} with a constant matrix where each element of the matrix is $\tau_i$. Assumption $\mathcal{A}_4$ claims that under high-dimensional setup, even when both $S_0$ and $SRCV$ are inconsistent, upon choosing the $\tau$'s, $S$ and $S_0$  have a closeness in the sense expressed by Eq.~\eqref{eq:closeness_condition} \& Eq.~\eqref{eq:closeness_condition2}.   \\
Now we state our main theorem.
\begin{theorem}\label{thm:main_theorem}
Under the assumptions ($\mathcal{A}$), almost surely, ESD of $S_{0}$ converges in distribution to a probability distribution with Stieltjes transform
\begin{equation}\label{eqn:thm1_1}
s(z)=-\frac{1}{z}\int\frac{1}{1-\frac{\lambda}{z}\int_{0}^{1}\frac{\tau'_{t}}{1+c\tau'_{t}\tilde{s}(z)}dt }dH(\lambda),
\end{equation}
 where $\tilde{s}(z)$ can be solved by the equation:
 \begin{equation}\label{eqn:thm1_2}
     \tilde{s}(z)=-\frac{1}{z}\int_{\lambda}\frac{\lambda}{1-\frac{\lambda}{z}\int_{0}^{1}\frac{\tau'_{t}}{1+c\tau'_{t}\tilde{s}(z)}dt}dH(\lambda).
 \end{equation}
\end{theorem}
This theorem establishes the Limiting spectral distribution of Hayashi-Yoshida estimator by determining its spectral distribution through the limiting spectral distribution of $\Sigma$. In other words the theorem establishes the link between the limiting spectral distributions of $\Sigma$ and Hayashi-Yoshida estimator through its Stieltjes transform.\\
See Appendix for the lemmas that would lead to the theorem.\\
Now we are ready to prove Theorem \ref{thm:main_theorem}.
\begin{proof}
Define \begin{equation*}S^*=\Sigma\sum_{l=1}^{n}\frac{\tau_{l}}{1+\tau_{l}a_{l}},\end{equation*}
where $a_{l}=Y_{l}'\Sigma^{\frac{1}{2}}(S_{l}-zI)^{-1}\Sigma^{\frac{1}{2}}Y_{l}$
and $S_{l}=\sum_{j\ne l}\Sigma^{\frac{1}{2}}Y_{j}Y_{j}'\Sigma^{\frac{1}{2}}$.
\textcolor{black}{We denote the Stieltjes transform of $F^{S}$ by
$s_{n}$,
\begin{equation*}
s_{n}:=s_{n}(z)=\frac{tr((S-zI)^{-1})}{p}.
\end{equation*}
}
Similarly, \textcolor{black}{Stieltjes transform of $F^{S_{0}}$and
Stieltjes transform of $F^{S^*}$ are denoted by $s_{n}^{0}$
and $s^*_{n}$ respectively.}\\
In order to show the convergence of $s_{n}(z)$ it is enough to show
for all $z=iv$ with $v>0$ sufficiently large and the condition,
in Lemma 8, is satisfied. \\
We will start by showing that \textcolor{black}{$s_{n}^{0}$ and $s^*_{n}$
}converge to the same limit.
\begin{align*}
s^*_{n}-s_{n}^{0} &=\frac{1}{p}tr[(S^*-zI)^{-1}-(S_{0}-zI)^{-1}]\\
 & =\frac{1}{p}tr[(S^*-zI)^{-1}(S-S^*)(S-zI)^{-1}]+\frac{1}{p}tr[(S-zI)^{-1}(S_{0}-S)(S_{0}-zI)^{-1}]
 \end{align*}
\begin{align*}
&tr[(S^*-zI)^{-1}(S-S^*)(S-zI)^{-1}] \\
&=tr[(\Sigma\sum_{l=1}^{n}\frac{\tau_{l}}{1+\tau_{l}a_{l}}-zI)^{-1}[\sum_{l=1}^{n}\tau_{l}\Sigma^{\frac{1}{2}}Y_{l}Y_{l}'\Sigma^{\frac{1}{2}}-\Sigma\sum_{l=1}^{n}\frac{\tau_{l}}{1+\tau_{l}a_{l}}](S-zI)^{-1}]\\
\\
&=tr[(\Sigma\sum_{l=1}^{n}\frac{\tau_{l}}{1+\tau_{l}a_{l}}-zI)^{-1}[\sum_{l=1}^{n}\tau_{l}\Sigma^{\frac{1}{2}}Y_{l}Y_{l}'\Sigma^{\frac{1}{2}}]\frac{(S_{l}-zI)^{-1}}{1+\tau_{l}a_{l}}]\\
 & -tr[(\Sigma\sum_{l=1}^{n}\frac{\tau_{l}}{1+\tau_{l}a_{l}}-zI)^{-1}[\Sigma\sum_{l=1}^{n}\frac{\tau_{l}}{1+\tau_{l}a_{l}}](S-zI)^{-1}]\\
 &=\sum_{l=1}^{n}\frac{\tau_{l}}{1+\tau_{l}a_{l}}(Y_{l}'\Sigma^{\frac{1}{2}}(S_{l}-zI)^{-1}(\Sigma\sum_{l=1}^{n}\frac{\tau_{l}}{1+\tau_{l}a_{l}}-zI)^{-1}\Sigma^{\frac{1}{2}}Y_{l}\\
 & -tr[\Sigma^{\frac{1}{2}}(S-zI)^{-1}(\Sigma\sum_{l=1}^{n}\frac{\tau_{l}}{1+\tau_{l}a_{l}}-zI)^{-1}\Sigma^{\frac{1}{2}}])
 \end{align*}\begin{align*}
tr[(S-zI)^{-1}(S_{0}-S)(S_{0}-zI)^{-1}] &\le \sqrt{tr(S_{0}-S)^{2}\times tr([(S-zI)^{-1}(S_{0}-zI)^{-1}]^{2})}
\end{align*}
\textcolor{black}{According to our assumption, $\mathcal{A}_6$, $\sum_{l=1}^{n}\frac{\tau_{l}}{1+\tau_{l}a_{l}}\leq n\ \mathrm{max}(\tau_{l})\leq\kappa$.}\\
It is enough to show that
\begin{align*}
\Xi & = \frac{1}{p}\mathrm{max_{l}}(Y_{l}'\Sigma^{\frac{1}{2}}(S_{l}-zI)^{-1}(\Sigma\sum_{l=1}^{n}\frac{\tau_{l}}{1+\tau_{l}a_{l}}-zI)^{-1}\Sigma^{\frac{1}{2}}Y_{l}\\
 & -tr[\Sigma^{\frac{1}{2}}(S-zI)^{-1}(\Sigma\sum_{l=1}^{n}\frac{\tau_{l}}{1+\tau_{l}a_{l}}-zI)^{-1}\Sigma^{\frac{1}{2}}])\rightarrow0\ \mathrm{a.s.}\end{align*}
 We write $\Xi=I+II+III+IV$ where \begin{align*}I=&\frac{1}{p}\mathrm{max_{l}}(Y_{l}'\Sigma^{\frac{1}{2}}(S_{l}-zI)^{-1}(\Sigma\sum_{l=1}^{n}\frac{\tau_{l}}{1+\tau_{l}a_{l}}-zI)^{-1}\Sigma^{\frac{1}{2}}Y_{l}\\
 & -Y_{l}'\Sigma^{\frac{1}{2}}(S_{l}-zI)^{-1}(\Sigma\sum_{j\neq l}\frac{\tau_{j}}{1+\tau_{j}b_{j}^{l}}-zI)^{-1}\Sigma^{\frac{1}{2}}Y_{l})
\end{align*}\begin{align*}II= &\frac{1}{p}\mathrm{max_{l}}(Y_{l}'\Sigma^{\frac{1}{2}}(S_{l}-zI)^{-1}(\Sigma\sum_{j\neq l}\frac{\tau_{j}}{1+\tau_{j}b_{j}^{l}}-zI)^{-1}\Sigma^{\frac{1}{2}}Y_{l}\\
 & -tr[\Sigma^{\frac{1}{2}}(S_{l}-zI)^{-1}(\Sigma\sum_{j\neq l}\frac{\tau_{j}}{1+\tau_{j}b_{j}^{l}}-zI)^{-1}\Sigma^{\frac{1}{2}}])
 \end{align*}\begin{align*}III=
 & \frac{1}{p}\mathrm{max_{l}}(tr[\Sigma^{\frac{1}{2}}(S_{l}-zI)^{-1}(\Sigma\sum_{j\neq l}\frac{\tau_{j}}{1+\tau_{j}b_{j}^{l}}-zI)^{-1}\Sigma^{\frac{1}{2}}]\\
 & -tr[\Sigma^{\frac{1}{2}}(S_{l}-zI)^{-1}(\Sigma\sum_{l=1}^{n}\frac{\tau_{l}}{1+\tau_{l}a_{l}}-zI)^{-1}\Sigma^{\frac{1}{2}}])
 \end{align*}\begin{align*}IV=
 &\frac{1}{p}\mathrm{max_{l}}(tr[\Sigma^{\frac{1}{2}}(S_{l}-zI)^{-1}(\Sigma\sum_{l=1}^{n}\frac{\tau_{l}}{1+\tau_{l}a_{l}}-zI)^{-1}\Sigma^{\frac{1}{2}}]\\
 & -tr[\Sigma^{\frac{1}{2}}(S-zI)^{-1}(\Sigma\sum_{l=1}^{n}\frac{\tau_{l}}{1+\tau_{l}a_{l}}-zI)^{-1}\Sigma^{\frac{1}{2}}])
\end{align*}
where $b_{j}^{l}=Y_{j}'\Sigma^{\frac{1}{2}}(S_{j,l}-zI)^{-1}\Sigma^{\frac{1}{2}}Y_{j}$ with $S_{j,l}=\sum_{i\neq j,l}\Sigma^{\frac{1}{2}}Y_{i}Y_{i}'\Sigma^{\frac{1}{2}}$.

It is sufficient to show that $\mathrm{I},\mathrm{II},\mathrm{III},\mathrm{IV}\rightarrow0\ \mathrm{a.s.}$

\textcolor{black}{Convergence of $\mathrm{I}$ and $\mathrm{II}$ are
followed by Lemma \ref{lemma1app} and Lemma \ref{lemma2app}. Convergence of $\mathrm{III}$ and
$\mathrm{IV}$ can be proved similarly.}

Now due to Assumption 5, we have
\begin{equation*}
\frac{1}{p}tr(\sum_{l=1}^{n}\frac{\tau_{l}}{1+\tau_{l}a_{l}}\Sigma-zI)^{-1}\rightarrow-\frac{1}{z}\int\frac{1}{1-\frac{\lambda}{z}\int_{0}^{1}\frac{\tau'_{t}}{1+c\tau'_{t}\tilde{s}(z)}dt}dH(\lambda).
\end{equation*}

As \begin{equation*}\frac{1}{p}tr(\sum_{l=1}^{n}\frac{\tau_{l}}{1+\tau_{l}a_{l}}\Sigma-zI)^{-1}-\frac{1}{p}tr(S-zI)^{-1}\rightarrow0,\end{equation*}
we get
\begin{equation}\label{eqn:sz}
s(z)=-\frac{1}{z}\int\frac{1}{1-\frac{\lambda}{z}\int_{0}^{1}\frac{\tau'_{t}}{1+c\tau'_{t}\tilde{s}(z)}dt}dH(\lambda).
\end{equation}

The fact that $\int_{0}^{1}\frac{\tau'_{t}}{1+c\tau'_{t}\tilde{s}(z)}dt\neq0$
and $Real(\int_{0}^{1}\frac{\tau'_{t}}{1+c\tilde{s}(z)}dt)\geq0$
implies $|s(z)|<\frac{1}{|z|}$ and therefore $1+zs(z)\neq0$. So
from Lemma 5, it is clear that $\tilde{s}(z)\neq0$.

Now
\begin{align*}
\int_{0}^{1}\frac{\tau'_{t}}{1+c\tilde{s}(z)\tau'_{t}}dt & =\frac{1}{c\tilde{s}(z)}(1-\int_{0}^{1}\frac{1}{1+c\tilde{s}(z)\tau'_{t}}dt)\\
 & =\frac{1}{c\tilde{s}(z)}(1-(1-c(1+zs(z))))\\
 & =\frac{1+zs(z)}{\tilde{s}(z)}
\end{align*}

Also, equation \ref{eqn:sz} implies
\begin{align*}
1+zs(z) & =-\frac{1}{z}\int_{0}^{1}\frac{\tau'_{t}}{1+c\tau'_{t}\tilde{s}(z)}dt\int_{\lambda}\frac{\lambda}{1-\frac{\lambda}{z}\int_{0}^{1}\frac{\tau'_{t}}{1+c\tau'_{t}\tilde{s}(z)}dt}dH(\lambda)\\
\implies\int_{0}^{1}\frac{\tau'_{t}}{1+c\tilde{s}(z)\tau'_{t}}dt & =-\frac{1}{z\tilde{s}(z)}\int_{0}^{1}\frac{\tau'_{t}}{1+c\tau'_{t}\tilde{s}(z)}dt\int_{\lambda}\frac{\lambda}{1-\frac{\lambda}{z}\int_{0}^{1}\frac{\tau'_{t}}{1+c\tau'_{t}\tilde{s}(z)}dt}dH(\lambda)\\
\implies\tilde{s}(z) & =-\frac{1}{z}\int_{\lambda}\frac{\lambda}{1-\frac{\lambda}{z}\int_{0}^{1}\frac{\tau'_{t}}{1+c\tau'_{t}\tilde{s}(z)}dt}dH(\lambda)
\end{align*}

Due to Lemma 6 and Lemma 7, $-\frac{1}{z}\sum_{l=1}^{n}\frac{\tau_{l}}{1+a_{l}\tau_{l}}\in Q_{1}$
and $\frac{1}{p}tr(S-zI)^{-1}\in Q_{1}$. Same will be true for their
limits. Now to show that $s(z)$ is not unique it is enough to show
that if there are two solutions $s_{1}(z)$ and $s_{2}(z)$ (and therefore
$\tilde{s}_{1}(z),\tilde{s}_{2}(z)$) then $s_{1}=s_{2}$. If possible
let there are two limiting spectral distributions $s_{1}$ and $s_{2}$
such that $s_{1}\neq s_{2}$. To show the contradiction, it is enough
to show $\tilde{s}_{1}(z)=\tilde{s}_{2}(z)$.

Note that,
\begin{equation*}
\int_{0}^{1}\frac{\tau'_{t}}{1+c\tau'_{t}\tilde{s}_{1}(z)}dt-\int_{0}^{1}\frac{\tau'_{t}}{1+c\tau'_{t}\tilde{s}_{2}(z)}dt
=\int_{0}^{1}\frac{c(\tau'_{t})^{2}(\tilde{s}_{2}(z)-\tilde{s}_{1}(z))}{(1+c\tau'_{t}\tilde{s}_{1}(z))(1+c\tau'_{t}\tilde{s}_{2}(z))}dt
\end{equation*}

But,

\begin{equation*}
\tilde{s}_{1}(z)-\tilde{s}_{2}(z)
=-\frac{1}{z}\int_{\lambda\in\mathbb{R}}\Bigg[\frac{\lambda}{1-\frac{\lambda}{z}\int_{0}^{1}\frac{\tau'_{t}}{1+c\tau'_{t}\tilde{s}_{1}(z)}dt}-\frac{\lambda}{1-\frac{\lambda}{z}\int_{0}^{1}\frac{\tau'_{t}}{1+c\tau'_{t}\tilde{s}_{2}(z)}dt}\Bigg]dH(\lambda)
\end{equation*}
On simplification, this gives
\begin{equation}
1=\frac{c}{z^{2}}\int_{0}^{1}\frac{(\tau'_{t})^{2}}{\big(1+c\tau'_{t}\tilde{s}_{1}(z)\big)\big(1+c\tau'_{t}\tilde{s}_{2}(z)\big)}dt\times\int_{0}^{1}\frac{\lambda^{2}}{\big(1-\frac{\lambda}{z}\int_{0}^{1}\frac{\tau'_{t}}{1+c\tau'_{t}\tilde{s}_{1}(z)}dt\big)\big(1-\frac{\lambda}{z}\int_{0}^{1}\frac{\tau'_{t}}{1+c\tau'_{t}\tilde{s}_{2}(z)}dt\big)}dH(\lambda)\label{1}
\end{equation}

As $\tilde{s}_{1},\tilde{s}_{2}\in Q_{1}$ , \begin{equation*}\Bigg|\int_{0}^{1}\frac{(\tau'_{t})^{2}}{\Big(1+c\tau'_{t}\tilde{s}_{1}(z)\Big)\Big(1+c\tau'_{t}\tilde{s}_{2}(z)\Big)}dt\Bigg|\leq\int_{0}^{1}(\tau'_{t})^{2}dt<\infty.\end{equation*}

And $-\frac{1}{z}\int_{0}^{1}\frac{\tau'_{t}}{1+c\tau'_{t}\tilde{s}_{i}(z)}dt$,
$i=1,2$ implies \begin{equation*}\Bigg|\int_{0}^{1}\frac{\lambda^{2}}{\Big(1-\frac{\lambda}{z}\int_{0}^{1}\frac{\tau'_{t}}{1+c\tau'_{t}\tilde{s}_{1}(z)}dt)(1-\frac{\lambda}{z}\int_{0}^{1}\frac{\tau'_{t}}{1+c\tau'_{t}\tilde{s}_{2}(z)}dt\Big)}dH(\lambda)\Bigg|\leq\int\lambda^{2}dH(\lambda)<\infty.\end{equation*}

So for $z=iv$, with $v$ sufficiently large, Eq. \eqref{1} can not be
true. So $s(z)$ is unique.
\end{proof}

The above theorem is true for time-varying instantaneous covolatility process with a little stronger set assumptions. Following \cite{zheng2011estimation}, we can assume that time varying covolatility process can be decomposed in two parts: a time varying \textit{cadlag} process and a symmetric matrix that is not varying with time. Formally, $\sigma_{t}=\gamma_{t}\tilde{\Sigma}^{\frac{1}{2}}$ where $\tilde{\Sigma}^{\frac{1}{2}}$ does not depend on time and as mentioned $\gamma_{t}$ is a time-varying cadlag process. \\
If we assume
this, then $\Delta X_{l}$ has the same distribution as the following-
\begin{equation*}
\Delta X_{l}=\tilde{\Sigma}^{\frac{1}{2}}\left(\begin{array}{c}
\int_{t_{l-1}^{1}}^{t_{l}^{1}}\gamma_{s}dW_{1}\\
\int_{t_{l-1}^{2}}^{t_{l}^{2}}\gamma_{s}dW_{2}\\
.\\
.\\
\int_{t_{l-1}^{p}}^{t_{l}^{p}}\gamma_{s}dW_{p}
\end{array}\right)\stackrel{d}{=}(\tilde{\Sigma}^{\frac{1}{2}}\circ A) D_l Y_{l},
\end{equation*}
where $Y_{l}^{p}=(Y_{l}^{j})_{1\leq j\leq p}$; $Y_{l}^{j}$ 's are iid normal with mean $0$ and variance $1$ for $p=1,2,...$ and $1\leq l\leq n$ and $A=(a_{ij})$ when $a_{ij}=\int_{t_{l-1}^{i}\vee t_{l-1}^{j}}^{t_{l}^{i}\wedge t_{l}^{j}}\gamma_{s}^{2}I((t_{l-1}^{i},t_{l}^{i})\cap I(t_{l-1}^{j},t_{l}^{j}))ds$, and\\ $D_l = diag(\sqrt{\int_{t_{l-1}^{1}}^{t_{l}^{1}}\gamma_{s}^{2}ds},\sqrt{\int_{t_{l-1}^{2}}^{t_{l}^{2}}\gamma_{s}^{2}ds},...,\sqrt{\int_{t_{l-1}^{p}}^{t_{l}^{p}}\gamma_{s}^{2}ds})$. \\
Therefore we are interested in the spectral distribution of
\begin{equation*}
S_{0}=\sum_{k}\sum_{l}\{(\Sigma^{\frac{1}{2}}\circ A)D_{k}x_{k}x_{l}^{T}D_{l}(\Sigma^{\frac{1}{2}}\circ A)\circ I(k,l)\}
\end{equation*}
 as both $S_{0}$ and $\Sigma_{p}^{HY}$ have the same LSD.\\
Suppose now we denote the Limiting spectral distribution of $\tilde{\Sigma}$ as $H$.  Then along with other assumptions of $\mathcal{A}$ we need the following additional assumptions
\begin{enumerate}
 \item[($\mathcal{A}_7$)] $\sigma_{t}=\gamma_{t}\tilde{\Sigma}^{\frac{1}{2}}$ where $\tilde{\Sigma}^{\frac{1}{2}}$ does not depend on time and $\gamma_{t}$ is a time-varying cadlag process.
    \item[($\mathcal{A}_8$)] $\int_{t_{l-1}^{j}}^{t_{l}^{j}}\gamma_{s}^{2}ds$ are independent
of $Y_{l}$.
\end{enumerate}
Then the result holds for time-varying covolatility process. How these conditions impose constraints on the time-varying covolatility or more specifically the cadlag process is a separate but interesting question to study. We write the theorem below
\begin{theorem}
Under the assumptions ($\mathcal{A}_1$-$\mathcal{A}_8$), almost surely, ESD of $S_{0}$ converges in distribution to a probability distribution with Stieltjes transform given by equations (\ref{eqn:thm1_1})-(\ref{eqn:thm1_2}).
\end{theorem}
\section{Spectral Analysis of Hayashi's Estimator when $p/n \rightarrow 0$}\label{sec:p/n_0}
Due to the fact that Hayashi's estimator is unbiased we will be concerned
about the following matrix (let us call $\tilde{S}$),
\begin{equation*}
\tilde{S}=\sqrt{\frac{n}{p}}(\frac{S_{0}}{n}-\frac{\Sigma_{p}}{n}).
\end{equation*}
 Like the previous chapter here also we have to take some assumptions,
we will call it \textbf{$\mathcal{B}$. }
\begin{enumerate}
\item[($\mathcal{B}$1)]: $\mathrm{lim}~\frac{p}{n}\rightarrow0$ as $p\rightarrow\infty$
and $n\rightarrow\infty$.

\item[($\mathcal{B}$2)]: $Z_{ij}'s$ ($1\leq i\leq p$, $1\leq j\leq n$) are iid Gaussian
random variables with $E(Z_{ij})=0$, $E|Z_{ij}|^{2}=1$.

\item[($\mathcal{B}$3)]: $F^{\Sigma_{p}}\stackrel{L}{\rightarrow}F^{H}(\neq\delta_{\{0\}})$
as $p\rightarrow\infty$ where $F^{H}$ is a distribution function.

\item[($\mathcal{B}$4)] : Define, $\bar{S}= \sqrt{\frac{n}{p}}(\frac{S}{n}-\frac{\Sigma_{p}}{n})$. Then $\bar{S}-\tilde{S}$, $(\tilde{S}-zI)^{-1}$ and $(\bar{S}-zI)^{-1}$ are positive definite, $tr(S_0-S)=O(p)$, $tr(\bar{S}-zI)^{-1}=O(1)$, $tr(\tilde{S}-zI)^{-1}=O(1)$ and $(1-n\sum_{l=1}^n \tau_l)tr(\Sigma)=O(p)$.
\item[($\mathcal{B}$5)] : $||\mathrm{diag}(\tau_{1},\tau_{2},...,\tau_{n})||$ is bounded
above.
\item[($\mathcal{B}$6)] : $\frac{1}{n}(\sum_{l=1}^{n}\tau_{l})\rightarrow\tau>0$ and $\frac{1}{n}(\sum_{l=1}^{n}\tau_{l}^{2})\rightarrow\bar{\tau}>0$
as $n\rightarrow\infty$
\end{enumerate}
Now we are ready to state the main theorem.
\begin{theorem}\label{thch5}
If the above assumptions ($\mathcal{B}$) are true
then the empirical spectral distribution of $\sqrt{\frac{n}{p}}(\frac{1}{n}S_{0}-\frac{1}{n}\Sigma_{p})$
almost surely converges weakly to a nonrandom distribution $F$ as
$n\rightarrow\infty$, whose Stieltjes transform $s(z)$ is determined
by the following system of equations:
\begin{equation*}
\begin{cases}
s(z)= & -\int\frac{dH(\lambda)}{z+\bar{\tau}\lambda\beta(z)}\\
\beta(z)= & -\int\frac{\lambda dH(\lambda)}{z+\bar{\tau}\lambda\beta(z)}
\end{cases}
\end{equation*}
for any $z\in\mathbb{C}_{+}$.
\end{theorem}
Proof of this theorem will be in the similar path as in \cite{wang2014limiting}. Before proving the theorem we will define some quantities and make
some observations.
Let $\Sigma=U^{*}\Lambda U$ be the spectral decomposition of $\Sigma$.
Now define,
\begin{equation*}
W=U\Sigma^{\frac{1}{2}}\{\sum_{l}\tau_{l}^{\frac{1}{2}}Z_{l}e_{l}^{*}\}.
\end{equation*}
Let $w_{k}$ be the $k$th row of $W$ and $W_{k}$ be the matrix
after deleting the $k$th row. Define,
\begin{align*}
h_{k}= & \sqrt{\frac{n}{p}}W_{k}w_{k},\\
M= & \sqrt{\frac{n}{p}}(WW^{*}-(\sum_{l}\tau_{l})\Lambda),\\
M_{k}= & \sqrt{\frac{n}{p}}(W_{k}W^{*}-(\sum_{l}\tau_{l})\Lambda_{k}),\\
\bar{M}_{k}= & \sqrt{\frac{n}{p}}(W_{k}W_{k}^{*}-(\sum_{l}\tau_{l})\Lambda_{k}),\\
t_{kk}= & \sqrt{\frac{n}{p}}(w_{k}w_{k}^{*}-(\sum_{l}\tau_{l})\lambda_{k}),
\end{align*}
where $\Lambda_{k}$ is the matrix obtained by deleting $k$th diagonal
element of $\Lambda$. \\
Now we will make some remarks. Justifications of the remarks are given
in the appendix.\\
\begin{enumerate}
    \item[Remark 1:] $M=\sum_{k=1}^{p}e_{k}(h_{k}+t_{kk}e_{k})^{*}$ and $M=M_{k}+e_{k}(h_{k}+t_{kk}e_{k})^{*}$.
    \item[Remark 2:] $tr(I+z(M-zI)^{-1})=\sum_{k=1}^{p}\big[(h_{k}+t_{kk}e_{k})^{*}(M_{k}-zI)^{-1}e_{k}\big]/\big[1+(h_{k}+t_{kk}e_{k})^{*}(M_{k}-zI)^{-1}e_{k}\big].$
    \item[Remark 3:] $(h_{k}+t_{kk}e_{k})^{*}(M_{k}-zI)^{-1}e_{k}=h_{k}^{*}\bar{(M}_{k}-zI)^{-1}h_{k}-t_{kk}/z.$
    \item[Remark 4:] $E(h_{k}^{*}\bar{(M}_{k}-zI)^{-1}h_{k}-(\sum_{l=1}^{n}\tau_{l}^{2})\lambda_{k}tr(\bar{(M}_{k}-zI)^{-1}\Lambda_{k})/np)=0.$
\end{enumerate}

Now we are ready to prove Theorem \ref{thch5}.
\begin{proof}[Proof of Theorem \ref{thch5}]
Suppose the Stiletjes transform of $\bar{S}_{n}$ is $\bar{s}_{n}(z)$. So $\bar{s}_{n}(z)=\frac{1}{p}\mathrm{tr}(\bar{S}_{n}-zI)^{-1}$.\\
Observe that, $z^{-1}I+(\bar{S}_{n}-zI)^{-1} =z^{-1}\bar{S}_{n}(\bar{S}_{n}-zI)^{-1}.$
This implies the following:
\begin{align}
(\bar{S}_{n}-zI)^{-1}= & z^{-1}\bar{S}_{n}(\bar{S}_{n}-zI)^{-1}-z^{-1}I\nonumber \\
\implies \hspace{.2cm} \bar{s}_{n}(z)= & -z^{-1}+\frac{z^{-1}}{p}tr(\bar{S}_{n}(\bar{S}_{n}-zI)^{-1}).\label{eq:1-1}
\end{align}
Notice that,
\begin{align*}
\bar{s}_{n}(z)-\tilde{s}_{n}(z) & =\frac{1}{p}[tr(\bar{S}_{n}-zI)^{-1}-tr(\tilde{S}_{n}-zI)^{-1}]\\
 & =\frac{1}{p}tr[(\tilde{S}_{n}-\bar{S}_{n})(\tilde{S}_{n}-zI)^{-1}(\bar{S}_{n}-zI)^{-1}]\\
 & \leq \frac{1}{p} tr(\tilde{S}_{n}-\bar{S}_{n})tr(\tilde{S}_{n}-zI)^{-1}tr(\bar{S}_{n}-zI)^{-1}\\
\end{align*}
Now,
\begin{align*}
tr(\tilde{S}_{n}-\bar{S}_{n}) & = tr\Big[\sqrt{\frac{n}{p}}\Big\{(S_0/n - S/n) - (\frac{1}{n}\Sigma - (\sum_{l=1}^n \tau_l)\Sigma)\Big\}\Big]\\
& = \sqrt{\frac{1}{np}}tr(S_0 - S) - \sqrt{\frac{1}{np}} \Big[1 - n(\sum_{l=1}^n \tau_l)tr(\Sigma)\Big]\\
& = o(p)~\text{a.s.}
\end{align*}
The last line of the above derivation is a consequence of Assumption
$\mathcal{B} 4$. Moreover as $tr(\bar{S}-zI)^{-1}=O(1)$ and $tr(\tilde{S}-zI)^{-1}=O(1)$, we have $\bar{s}_{n}(z)-\tilde{s}_{n}(z)\rightarrow 0, $ a.s.\\
This means that we can derive valuable
information about $\tilde{s}_{n}(z)$ by studying the spectral distribution
of $\bar{S}_{n}(z)$. \\
Note that, according to our definitions
\begin{align}
tr(\bar{S}_{n}-zI)^{-1} & =tr[\sqrt{\frac{n}{p}}(\Sigma^{\frac{1}{2}}\{\sum_{l}\tau_{l}^{\frac{1}{2}}Y_{l}e_{l}^{*}\})(\Sigma^{\frac{1}{2}}\{\sum_{l}\tau_{l}^{\frac{1}{2}}Y_{l}e_{l}^{*}\}^{*}-(\sum_{l}\tau_{l})\Sigma))-zI]^{-1}\nonumber \\
 & =tr[\sqrt{\frac{n}{p}}U^{*}(U\Sigma^{\frac{1}{2}}\{\sum_{l}\tau_{l}^{\frac{1}{2}}Y_{l}e_{l}^{*}\})(\{\sum_{l}\tau_{l}^{\frac{1}{2}}Y_{l}e_{l}^{*}\}^{*}\Sigma^{\frac{1}{2}}U^{*}U-(\sum_{l}\tau_{l})U^{*}\Lambda U))-zI]^{-1}\nonumber \\
 & =tr[\sqrt{\frac{n}{p}}(WW^{*}-(\sum_{l}\tau_{l})\Lambda)-zI]^{-1}.\label{eq:2-1}
\end{align}
Now from Eq.~\eqref{eq:1-1} we have,
\begin{align*}
\bar{s}_{n}(z)
 & =-z^{-1}+\frac{z^{-1}}{p}tr(\sqrt{\frac{n}{p}}(WW^{*}-(\sum_{l}\tau_{l})\Lambda)(\sqrt{\frac{n}{p}}((WW^{*}-(\sum_{l}\tau_{l})\Lambda)-zI)^{-1})\\
 & =-z^{-1}+\frac{z^{-1}}{p}tr\{M(M-zI)^{-1})\\
 & =-\frac{z^{-1}}{p}\sum_{k=1}^{p}\frac{1}{1+(h_{k}+t_{kk}e_{k})^{*}(M_{k}-zI)^{-1}e_{k}}\\
 & =-\frac{1}{p}\sum_{k=1}^{p}\frac{1}{z-t_{kk}+\frac{n\lambda_{k}(\sum_{l}\tau_{n}^{2})}{p}tr(\Lambda_{k}\bar{(M}_{k}-zI)^{-1})+\epsilon_{1k}},
\end{align*}

Define $\beta_{n}=\frac{1}{p}tr(\bar{S}_{n}-zI)^{-1}\Sigma$, then
similar derivation will lead to
\begin{equation}
\beta_{n}(z)=-\frac{1}{p}\sum_{k=1}^{p}\frac{\lambda_{k}}{z+h_{k}^{*}\bar{(M}_{k}-zI)^{-1}h_{k}-t_{kk}}.\label{eq:3}
\end{equation}
But again,
\begin{align*}
\beta_{n}(z) & =\frac{1}{p}tr(\bar{S}_{n}-zI)\Sigma)^{-1}\\
\implies\beta_{n}(z) & =\frac{1}{p}tr(\bar{S}_{n}-zI)U^{*}\Lambda U)^{-1}\\
\implies\beta_{n}(z) & =\frac{1}{p}tr(M-zI)^{-1}\Lambda)\qquad\mathrm{by~Eq.}\ \eqref{eq:2-1}\\
\implies\lambda_{k}n(\sum_{l=1}^{n}\tau_{l}^{2})\beta_{n}(z) & =\lambda_{k}n(\sum_{l=1}^{n}\tau_{l}^{2})\frac{1}{p}tr(\Lambda(M-zI)^{-1}).
\end{align*}
By Remark 4, we know that $h_{k}^{*}\bar{(M}_{k}-zI)^{-1}h_{k}=n(\sum_{l=1}^{n}\tau_{l}^{2})\lambda_{k}tr(\Lambda_{k}\bar{(M}_{k}-zI)^{-1})/p+\epsilon_{1k}$.
Define, $\epsilon_{2k}=\lambda_{k}n(\sum_{l=1}^{n}\tau_{l}^{2})E(\beta_{n}(z))-h_{k}^{*}\bar{(M}_{k}-zI)^{-1}h_{k}$.
Therefore from Eq.~\eqref{eq:3},
\begin{align*}
\beta_{n}(z)  & =-\frac{1}{p}\sum_{k=1}^{p}\frac{\lambda_{k}}{z+\lambda_{k}n(\sum_{l=1}^{n}\tau_{l}^{2})E(\beta_{n}(z))}+\\
 & \{-\frac{1}{p}\sum_{k=1}^{p}\frac{\lambda_{k}}{z+\lambda_{k}n(\sum_{l=1}^{n}\tau_{l}^{2})E(\beta_{n}(z))-t_{kk}-\epsilon_{2k}}+\frac{1}{p}\sum_{k=1}^{p}\frac{\lambda_{k}}{z+\lambda_{k}n(\sum_{l=1}^{n}\tau_{l}^{2})E(\beta_{n}(z))}\}.
\end{align*}
Let us consider the second part of the right hand side. This equals
\begin{align*}
 \frac{1}{p}\sum_{k=1}^{p}\frac{-\lambda_{k}(t_{kk}+\epsilon_{2k})}{(z+\lambda_{k}n(\sum_{l=1}^{n}\tau_{l}^{2})E(\beta_{n}(z)))^{2}}\\
 & +\frac{1}{p}\sum_{k=1}^{p}\frac{-\lambda_{k}(t_{kk}+\epsilon_{2k})}{(z+\lambda_{k}n(\sum_{l=1}^{n}\tau_{l}^{2})E(\beta_{n}(z)))(z+\lambda_{k}n(\sum_{l=1}^{n}\tau_{l}^{2})E(\beta_{n}(z))-t_{kk}-\epsilon_{2k})}\\
 & -\frac{1}{p}\sum_{k=1}^{p}\frac{-\lambda_{k}(t_{kk}+\epsilon_{2k})}{(z+\lambda_{k}n(\sum_{l=1}^{n}\tau_{l}^{2})E(\beta_{n}(z)))^{2}}\\
 & =-\frac{1}{p}\sum_{k=1}^{p}\frac{\lambda_{k}(t_{kk}+\epsilon_{2k})}{(z+\lambda_{k}n(\sum_{l=1}^{n}\tau_{l}^{2})E(\beta_{n}(z)))^{2}}\\
 & +\frac{1}{p}\sum_{k=1}^{p}\lambda_{k}(\epsilon_{2k}+t_{kk})\frac{-(t_{kk}+\epsilon_{2k})}{(z+\lambda_{k}n(\sum_{l=1}^{n}\tau_{l}^{2})E(\beta_{n}(z)))^{2}(z+\lambda_{k}n(\sum_{l=1}^{n}\tau_{l}^{2})E(\beta_{n}(z))-t_{kk}-\epsilon_{2k})}\\
 & =\epsilon_{3k}+\epsilon_{4k}.
\end{align*}
We try to argue that $|E(\epsilon_{2k})|\rightarrow0$ as $p\rightarrow\infty$,
\begin{align*}
|E(\epsilon_{2k})| & =|E(\lambda_{k}\frac{(\sum_{l=1}^{n}\tau_{l}^{2})}{n}E(\beta_{n}(z))-h_{k}^{*}\bar{M}_{k}(z)^{-1}h_{k})|\\
 & =|\lambda_{k}\frac{\sum_{l=1}^{n}\tau_{l}^{2}}{np}E\{tr((M(z)^{-1}\Lambda)-tr(\bar{M}_{k}(z)^{-1}(\Lambda-\lambda_{k}e_{k}e_{k}^{*}))\}-\epsilon_{1k}|\\
 & =|\lambda_{k}\frac{\sum_{l=1}^{n}\tau_{l}^{2}}{np}E\{tr((M(z)^{-1}-\bar{M}_{k}(z)^{-1})\Lambda)-\frac{\lambda_{k}}{z})\}-\epsilon_{1k}|\\
 & \leq\lambda_{k}\frac{\sum_{l=1}^{n}\tau_{l}^{2}}{np}|Etr((M(z)^{-1}-\bar{M}_{k}(z)^{-1})\Lambda)|+|\frac{\lambda_{k}^{2}(\sum_{l=1}^{n}\tau_{l}^{2})}{znp}|+\lambda_{k}\frac{\sum_{l=1}^{n}\tau_{l}^{2}}{np}|\epsilon_{1k}|.
\end{align*}
It can be shown that the above expression is $o(p)$ (see \cite{wang2014limiting}).
\begin{align*}
E(|\epsilon_{2k}+t_{kk}|^{2}) & =E|\epsilon_{2k}+t_{kk}-E(\epsilon_{2k})|^{2}+|E(\epsilon_{2k})|^{2}\\ & =E|-h_{k}^{*}\bar{M}_{k}(z)^{-1}h_{k}+t_{kk}+(\sum_{l=1}^{n}\tau_{l}^{2})\lambda_{k}trE(\Lambda_{k}\bar{M}_{k}(z)^{-1})/np|^{2}+|E(\epsilon_{2k})|^{2}\\
 & =E|-h_{k}^{*}\bar{M}_{k}(z)^{-1}h_{k}+t_{kk}+(\sum_{l=1}^{n}\tau_{l}^{2})\lambda_{k}trE(\Lambda_{k}\bar{M}_{k}(z)^{-1})/np|^{2}\\
 & +E|(\sum_{l=1}^{n}\tau_{l}^{2})\lambda_{k}trE(\Lambda_{k}\bar{M}_{k}(z)^{-1})/np-(\sum_{l=1}^{n}\tau_{l}^{2})\lambda_{k}tr(\Lambda_{k}\bar{M}_{k}(z)^{-1})/np|^{2}+|E(\epsilon_{2k})|^{2}.
\end{align*}
It can be shown that all three terms is $o(p)$. \\
We further observe that
\[
|z+\lambda_{k}\frac{(\sum_{l=1}^{n}\tau_{l}^{2})}{n}E(\beta_{n}(z))|\geq Im(z+\lambda_{k}\frac{(\sum_{l=1}^{n}\tau_{l}^{2})}{n}E(\beta_{n}(z)))\geq Im(z)=v>0
\]
and
\[
|z+\lambda_{k}\frac{(\sum_{l=1}^{n}\tau_{l}^{2})}{n}E(\beta_{n}(z))-t_{kk}-\epsilon_{2k}|\geq|z+h_{k}^{*}\bar{M}_{k}^{-1}h_{k}-t_{kk}|\geq v.
\]
This implies $\epsilon_{3k}\rightarrow0$ a.s. \\
With this we showed that
\[
E(\beta_{n}(z))=-\frac{1}{p}\sum_{k=1}^{p}\frac{\lambda_{k}}{z+\lambda_{k}\frac{(\sum_{l=1}^{n}\tau_{l}^{2})}{n}E(\beta_{n}(z))}+d_{n},
\]
where $d_{n}\rightarrow0$ as $n\rightarrow\infty$. \\
If we replace $(x_{1i},...,x_{pi})$ by $(x'_{1i},...,x'_{pi})$ and
call the resulting new $\bar{s}_{n}$ by $s'_{n}$ (similarly $\beta_{n}$
by $\beta'_{n}(z)$) then it is easy to show that
\[
|\bar{s}_{n}(z)-s'_{n}(z)|\leq\frac{c_{1}}{pv}
\]
and
\[
|\bar{\beta}_{n}(z)-\beta'_{n}(z)|\leq\frac{c_{1}}{pv}.
\]
Hence by Lemma 9,
\[
\sum_{n}P(|\bar{s}_{n}(z)-s'_{n}(z)|>\delta)<\infty
\]
and
\[
\sum_{n}P(|\bar{\beta}_{n}(z)-\beta'_{n}(z)|>\delta)<\infty.
\]
By Borel-Cantelli lemma, $\bar{s}_{n}(z)-E(\bar{s}_{n})\rightarrow0$
a.s. and $\beta_{n}(z)-E(\beta_{n}(z))\rightarrow0$ a.s. and thus
we also have $\tilde{s}_{n}(z)-E(\bar{s}_{n})\rightarrow0$.
Now as $\text{max}_{k}|\lambda_{k}|\leq a_{0}$ and $|z+\lambda_{k}n(\sum_{l=1}^{n}\tau_{l}^{2})E(\beta_{n}(z))|\geq v$
is bounded, $E(\beta_{n}(z))$ is bounded. So by dominated
convergence theorem $E(\beta_{n}(z))$ converges to $\beta(z)$. But
as $\beta_{n}(z)-E(\beta_{n}(z))\rightarrow0$ a.s., we have $\beta_{n}(z)\rightarrow\beta(z)$
a.s. Similarly for $s_{n}(z)\rightarrow s(z)$ a.s. So $s_{n}(z)$
can be evaluated by the following two equations,
\[
s(z)=-\int\frac{dH^{\Sigma}(\lambda)}{z+n(\sum_{l}\tau_{l}^{2})\lambda E(\beta(z))}
\]
\\
\[
\beta(z)=-\int\frac{\lambda dH^{\Sigma}(\lambda)}{z+n(\sum_{l}\tau_{l}^{2})\lambda E(\beta(z))}.
\]
\end{proof}

\section{Data Analysis}\label{sec:data analysis}
\subsection{Simulated Data Analysis}\label{sec:Simulation}
Although we are working with a high-dimensional set up, the computational
complexity of the Hayashi's estimator is worth paying attention to.
The fact that the time to compute the Hayashi's estimator is much greater
compared to the Realized covariance estimator, restricts us to a moderate
dimension and sample size in the simulation study.
We simulated 30 stocks with each 500 observations where the spot volatility
matrix is taken to be $I$. The empirical cdf of Hayashi-Yoshida estimator and the cdf of integrated covariance matrix are shown in Fig. \ref{fig:theoretical_empirical_cdf}. The red line is obtained by generating data from the same underlying process on sufficiently fine synchronous grid and calculating the realized covariance for such data. It serves as the proxy of the spectrum of Integrated covariance matrix.  One limitation of the simulation study is that we have taken same number of observations for each stock. This is of course not a practical assumption but as discussed in Sec.\ref{sec:HY_n_refresh_time} it corresponds to the refresh time sampling.

We create the nonsynchronous data using the following algorithm. \\

\begin{center}
\fbox{\begin{minipage}{0.99\textwidth}
\textbf{Algorithm}
\begin{enumerate}
    \item Initialise $p$ (the number of stocks), $n_{1}, n_{2}, n_{3},..., n_{p}$ (the number of observations in each stock) and $D$ (the interval $[0,D]$ represents a day).
    \item Draw a sample of size $n=\sum n_{i}$ from a uniform distribution on $[0,D]$ where $[0,D]$ represents a day. Denote it by $T=\{t_{1},...,t_{n}\}$. Assume $t_{0}=0$.
    \item Generate $n$ random vectors ($\mathbf{x}_{j}, j=1(1)n$) from a $p$-dimensional distribution. Denote them by $\mathbf{x}_{j}, j=1(1)n$. Scale the $\mathbf{x}_{j}$'s appropriately to represent the increment in returns for the interval ($t_{j}-t_{j-1}$): $\mathbf{x}_{j}=\thinspace\mathbf{x}_{j}\sqrt{t_{j}-t_{j-1}}$.
    \item Define $\mathbf{y}_{k}=\sum_{1}^{k}\mathbf{x_{j}} \quad\forall k\in \{1,2,...,n\}$
    \item for $i=1(1)p$

        \quad From $T$ take a random subset of size $n_{i}$. Denote it by $T_{i}$.

        \quad Data for $i$th stock is \{($x_{i,j}, t_{j}), j=1(1)n_{i}$\}

        \quad Update $T=T-T_{i}$ i.e. removing the time points chose for $T_{i}$ from $T$.

\end{enumerate}
\end{minipage} }
\end{center}

\begin{figure}
\centering
\includegraphics[scale=0.6]{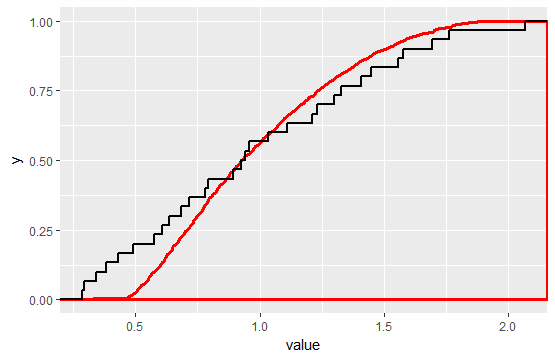}
\caption{Plot for empirical cumulative distribution function of Hayashi-Yoshida estimator and
cumulative distribution functions of Integrated covariance matrix where the spot volatility matrix is taken to be I for 30 stocks and 500 observations.}\label{fig:theoretical_empirical_cdf}
\end{figure}

We want to see the effect of $\gamma$, that is the ratio of $p/n$ on the empirical spectral distribution. So we repeat the simulation with $p=30$ and $n=60$. The result is given in the left panel of Fig.~\ref{fig:two simulations}

Next, we create a similar plot when the stocks are dependent. We have taken a 30-dimensional positive semi-definite covariance matrix with $p = 30$ stocks. As nontrivial high-dimensional covariance matrix is difficult to prefix, we take the $30\times30$ principal sub-matrix of the estimated covariance matrix from real data (Sec.~\ref{realdata}). The right panel of Fig.~\ref{fig:two simulations} shows the distribution of the eigenvalues. We can see that for general covariance matrix Hayashi-Yoshida estimator may not be positive definite.

\begin{figure}
     \centering
     \includegraphics[scale=0.4]{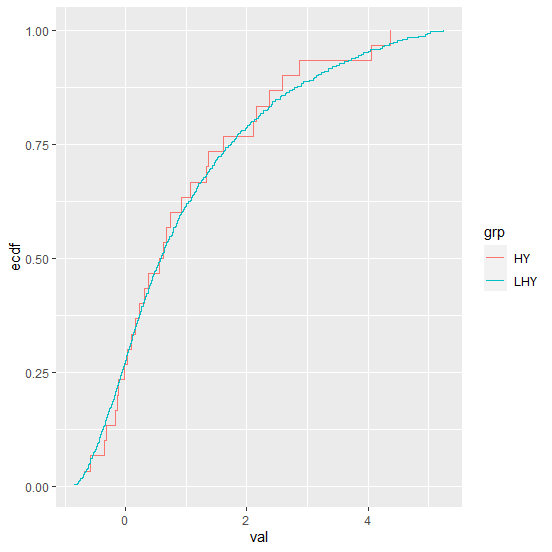}\includegraphics[scale=0.4]{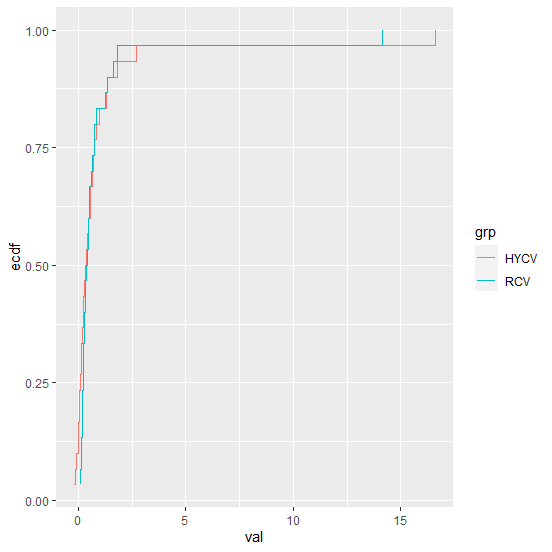}
        \caption{Plot for empirical cumulative distribution function of Hayashi-Yoshida estimator and
cumulative distribution functions of Integrated covariance matrix. Left: The spot volatility matrix is taken to be $I$ for 30 variables and sample size 60. Right: Spot volatility matrix $\Sigma$ is different from $I$ with 30 variables and sample size 500.}
        \label{fig:two simulations}
\end{figure}

It is clear from the algorithm that the time points are generated by a Poisson Process with different intensity parameter. For computational convenience we have taken all $n_{i}$'s to be same ($=n$).

\subsection{Real Data analysis}\label{realdata}
The limiting spectral distribution is particularly useful to test for deviation from null model, for example, whether the covolatility process is $I$ or not. Spectrum of integrated covariance matrix also helps to understand some of the key properties of the interacting units of the intraday-financial-network \cite{kumar2012correlation}. The extreme (highest) eigenvalue, for example, gives us significant insight about the \textit{market mode} or the collective overall response of the market to some external information. Spectral analysis, therefore, reveals broadly three types of fluctuations: (i) common to all stocks (i.e., due to market), (ii) related to a particular business sector (e.g. sectoral) and (iii) limited to an individual stock (i.e., idiosyncratic). These can be captured by simply segregating the network spectrum into the following parts: (i) the extreme eigenvalue (ii) eigenvalues deviating from the theoretical spectral distribution and (iii) bulk of the spectrum (\cite{plerou1999universal}, \cite{sinha2010econophysics}, \cite{onatski2010determining}). Limiting spectral distribution of Hayashi-Yoshida estimator would help us to identify the sectoral mode of intraday financial network.

We collect intraday tick by tick Bloomberg data of equities in Nifty 50 for several days. Here we present the results for three consecutive days starting from 22-12-2020 which are fairly representative.
In Fig.\ref{Fig:scree_plots}, we have plotted the scree plots of eigenvalues for these 50 stocks for the three days on the left panel. We see that the impact of the market mode makes the largest eigenvalue away from the bulk. On the right panel some of the eigenvectors, for the corresponding days, are plotted. Specifically, these are the eigenvector 3 of day 1, eigenvector 2 of day 2 and eigenvector 3 of day 3. Each of these has high contributions from stocks 2,4,13,14, 36 with same sign. These stocks are all from IT sector and there are no other stocks from IT sector in our dataset. This suggests that the IT sector mean (same sign) is the next big component that drives the market after the overall mean (market mode).
\begin{figure}\begin{center}
\includegraphics[scale=0.4]{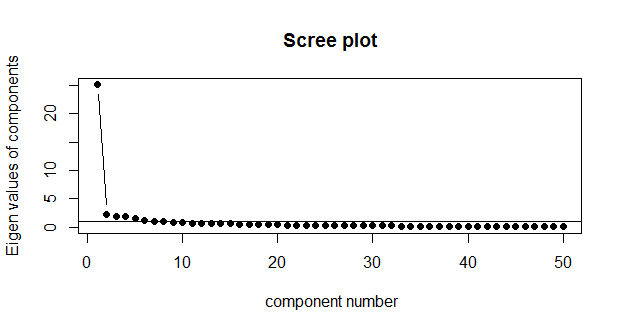}\includegraphics[scale=0.4]{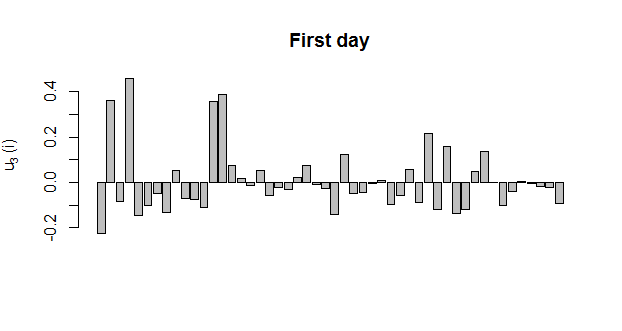}\\
\includegraphics[scale=0.4]{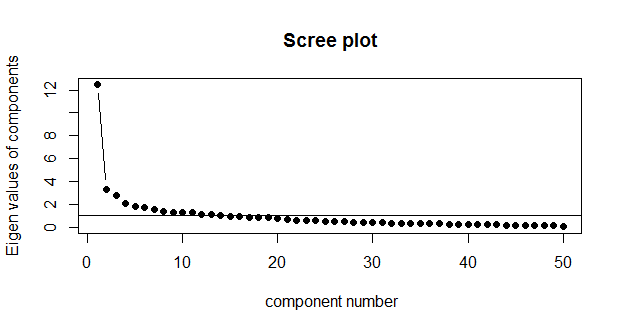}\includegraphics[scale=0.4]{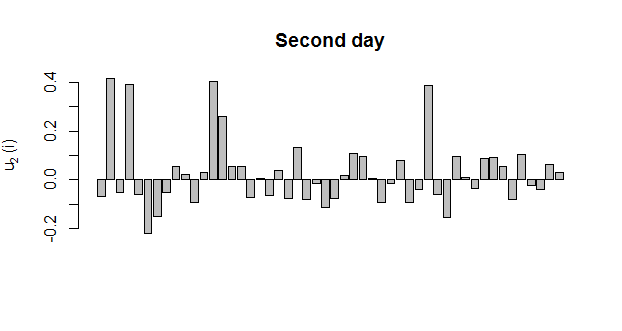}\\
\includegraphics[scale=0.4]{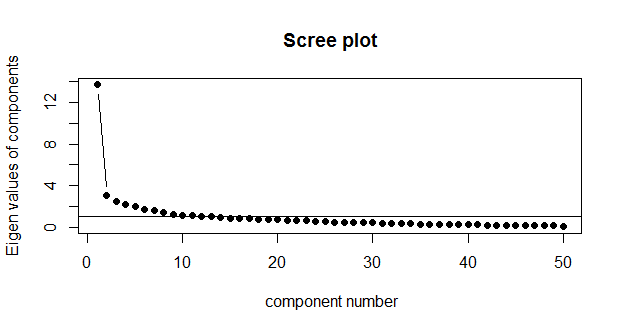}\includegraphics[scale=0.4]{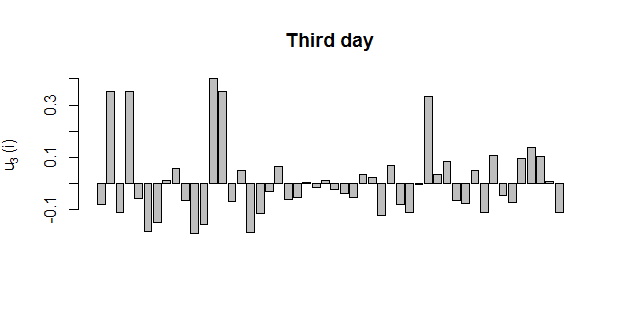}
\caption{Left panel: Scree plots of three consecutive days are plotted. We can see that the highest eigenvalue representing the market mode is away from the bulk. Right panel: Third eigenvector obtained from the data of first day, second eigenvector of second day and third eigenvector of the third day.}\label{Fig:scree_plots}
		\end{center}
	\end{figure}

\section{Conclusion and further directions}\label{sec:conclusion}
In this work we have determined the limiting spectral distribution of high dimensional Hayashi-Yoshida estimator for nonsynchronous intraday data.  Limiting spectral distribution can help to construct a shrinkage estimator of high-dimensional integrated covariance matrix (see \cite{ledoit2012nonlinear}). It can also be used for testing for a particular structure in spot volatility matrix.

In this paper we have only considered asynchronicity but not the presence of microstructure noise, which is also a feature of intraday stock-price data. So a natural direction to extend this work is by adding microstructure noise to it. Significant insights can be obtained from \cite{xia2018inference} where the same was derived for realized (co)volatility matrix. In presence of noise the spectral distribution may deviate from the ideal situation in significant ways. We have restricted ourselves to the simple Black-Scholes setup. Geometric Brownian motion models are not always very realistic models to describe financial data. One can try to go beyond that and investigate into the spectral analysis of Hayashi's estimator for more complex models. One can also try to extend the results for a general class of time varying covolatility processes (for more details, see \cite{zheng2011estimation, xia2018inference}). Changes due to leverage effects can also be quite serious and so worth looking into.

\bibliography{mainbib}

\appendix{Appendix}
\section{Lemmas}
\begin{lemma}
\textcolor{black}{Let $w_{1,}w_{2}\in\mathbb{C}$, with $Re(w_{1})\geq0$
and $Re(w_{2})\geq0$, $A$ is a $p\times p$ Hermitian nnd matrix,
$B$ being any $p\times p$ matrix, and $q\in\mathbb{C^{p}}$, then
\[
|q^{*}B(w_{1}A+I)^{-1}q-q^{*}B(w_{2}A+I)^{-1}q|\leq|w_{1}-w_{2}|~|q|^{2}~\|B\|~ \|A\|.
\]
}
\end{lemma}
For proof, see \cite{bai1999methodologies}.
\begin{lemma}
Let $z\in\mathbb{C}$, with $v=I(z)>0$, $A$ and $B$ are $p\times p$ matrices, with $B$ being Hermitian, and $q\in\mathbb{C}^{p}$. Then
\[
|tr(((B-zI)^{-1}-(B+\theta qq^{*}-zI)^{-1}A|\leq\frac{||A||}{v}
\]
for all $\theta \in\mathbb{R}.$
\end{lemma}
For proof, see \cite{bai1999methodologies}.
\begin{lemma}
For any Hermitian matrix $A$ and $z\in\mathbb{C}$, with $Im(z)=v>0$-
\[
||(A-zI)^{-1}||\leq1/v.
\]
\end{lemma}
For proof, see \cite{bai2010spectral}.
\begin{lemma}
For $X=(X_{1},X_{2},...,X_{p})^{T}$ where $X_{j}$'s are iid random
variables such that $\mathbb{E}(X_{1})=0$ , $\mathbb{E}|X_{1}|^{2}=1$,
and $\mathbb{E}|X_{1}|^{2k}<\infty$ for some $2\leq k\in\mathbb{N}$,
there exists $C_{k}\geq0$, depending only on $k$, $\mathbb{E}|X_{1}|^{4}$and
$\mathbb{E}|X_{1}|^{2k}$, such that for any $p\times p$ nonrandom
matrix $A$,
\begin{equation}
\mathbb{E}|X^{*}AX-tr(A)|^{2k}\leq C_{k}(tr(AA^{*}))^{k}\leq C_{k}p^{k}\|A\|^{2k}\label{eq:11}.
\end{equation}
\end{lemma}
For proof, see \cite{zheng2011estimation}.
\begin{lemma}
Suppose $S$ is a matrix defined as assumption $\mathcal{A}_{4}$, $z\in \mathbb{C}$ and $\frac{1}{p} \text{tr}(S_{n_{k}}-zI)^{-1}\rightarrow s(z)$, then
\begin{equation*}
    \sum_{l=1}^{n_{k}}\frac{\tau_{l}}{1+\tau_{l}a_{l}}\rightarrow\int_{0}^{1}\frac{\tau'_{t}}{1+c\tau'_{t}\tilde{s}(z)}dt \neq 0
\end{equation*}
 where $a_{l}=Y_{l}'\Sigma^{\frac{1}{2}}\big(\sum_{j\ne l}\Sigma^{\frac{1}{2}}Y_{j}Y_{j}'\Sigma^{\frac{1}{2}}-zI\big)^{-1}\Sigma^{\frac{1}{2}}Y_{l}$ ($Y_l$ is defined as in assumption $\mathcal{A}_3$) and $\tilde{s}(z)$ is the unique solution in $Q_{1}$ to the following equation:
\[
\int_{0}^{1}\frac{\tau'_{t}}{1+c\tau'_{t}\tilde{s}(z)}dt=1-c(1+zs(z)).
\]
\end{lemma}
For proof, see \cite{zheng2011estimation}.
\begin{lemma}
Let $z=iv\in\mathbb{C},$with $v>0$, $A$ be any $p\times p$ matrix,
and $B$ be a $p\times p$ Hermitian nonnegative definite matrix.
Then $tr(A(B-zI)^{-1}A^{*})\in Q_{1}$.
\end{lemma}
For proof, See \cite{zheng2011estimation}.
\begin{lemma}
Let $z=iv\in\mathbb{C}$ with $v>0$, A be a $p\times p$ Hermitian
nonnegative definite matrix, $q\in\mathbb{C}^{p},a>0$. Then
\begin{equation}
-\frac{1}{z}.\frac{1}{1+a.q^{*}(A-zI)^{-1}q}\in Q_{1}=\{z\in\mathbb{C}:\thinspace Re(z)\geq0,\thinspace Im(z)\geq0\}.\label{eq:12}
\end{equation}
\end{lemma}
For proof, see \cite{zheng2011estimation}.
\begin{lemma}
Suppose that $P_{n}$ are real probability measures with Stieltjes
transforms $s_{n}(z)$ . Let $K\subset\mathbb{C}_{+}$ be an infinite
set with limit points in $\mathbb{C}_{+}$. If $lims_{n}(z)=s(z)$
exists for all $z\in K$, then there exists a Probability measure
$P$ with Stieltjes transform $m(z)$ if and only if
\[
lim_{v\rightarrow\infty}iv.s(iv)=-1
\]
in which case $P_{n}\rightarrow P$.
\end{lemma}
For proof, see \cite{geronimo2003necessary}.

The next Lemma is known as McDiarmid Inequality \cite{mcdiarmid1997centering}.
\begin{lemma}
Let $Y_{1},Y_{2},...,Y_{m}$ be independent random vectors taking
values in $\mathcal{X}$. Suppose that $f:\mathcal{X}^{k}\rightarrow\mathbb{R}$
is a function of $Y_{1},Y_{2},...,Y_{m}$ satisfying $\forall y_{1},y_{2},..,y_{m},y_{i}'$,
\[
|f(y_{1},y_{2},..,y_{i},..,y_{m})-f(y_{1},y_{2},..,y_{i}',..,y_{m})|\leq c_{i},
\]
Then for all $\epsilon>0,$
\[
P(|f(y_{1},y_{2},...,y_{m})-f(y_{1},y_{2},...,y_{m})|>\epsilon)\leq2exp(-\frac{2\epsilon^{2}}{\sum_{i=1}^{m}c_{i}^{2}}).
\]
\end{lemma}
\begin{lemma}\label{lemma1app}\footnote{Lemma \ref{lemma1app} and \ref{lemma2app} can be found in Theorem 1 in \citep{zheng2011estimation} in slightly different form.}
For $Y_{l}$, $\Sigma$, $\tau_{l}$ described as Assumption $\mathcal{A}$
and $z=iv$, define,
\[
M_{1}=Y_{l}'\Sigma^{\frac{1}{2}}(S_{l}-zI)^{-1}(\Sigma\sum_{l=1}^{n}\frac{\tau_{l}}{1+\tau_{l}a_{l}}-zI)^{-1}\Sigma^{\frac{1}{2}}Y_{l}
\]
and
\[
M_{2}=Y_{l}'\Sigma^{\frac{1}{2}}(S_{l}-zI)^{-1}(\Sigma\sum_{j\neq l}\frac{\tau_{j}}{1+\tau_{j}b_{j}^{l}}-zI)^{-1}\Sigma^{\frac{1}{2}}Y_{l}
\]
Then,
\[
\frac{1}{p}\mathrm{max_{l}}(M_{1}-M_{2})\rightarrow0\ a.s.
\]
\end{lemma}

\begin{proof}[Proof of Lemma \ref{lemma1app}]

\begin{align*}
 & \frac{1}{p}\mathrm{max_{l}}|Y_{l}'\Sigma^{\frac{1}{2}}(S_{l}-zI)^{-1}(\Sigma\sum_{l=1}^{n}\frac{\tau_{l}}{1+\tau_{l}a_{l}}-zI)^{-1}\Sigma^{\frac{1}{2}}Z_{l}-Y_{l}'\Sigma^{\frac{1}{2}}(S_{l}-zI)^{-1}(\Sigma\sum_{j\neq l}\frac{\tau_{j}}{1+\tau_{j}b_{j}^{l}}-zI)^{-1}\Sigma^{\frac{1}{2}}Y_{l}|\\
 & =\frac{1}{p}\mathrm{max_{l}}|Y_{l}'\Sigma^{\frac{1}{2}}(S_{l}-zI)^{-1}(-\frac{1}{z})(-\frac{1}{z}\Sigma\sum_{l=1}^{n}\frac{\tau_{l}}{1+\tau_{l}a_{l}}+I)^{-1}\Sigma^{\frac{1}{2}}Y_{l}\\
 & -Y_{l}'\Sigma^{\frac{1}{2}}(S_{l}-zI)^{-1}(-\frac{1}{z})(-\frac{1}{z}\Sigma\sum_{j\neq l}\frac{\tau_{j}}{1+\tau_{j}b_{j}^{l}}+I)^{-1}\Sigma^{\frac{1}{2}}Y_{l}|\\
 & \leq\frac{1}{p}\mathrm{max}_{l}|-\frac{1}{z}\sum_{l}\frac{\tau_{l}}{1+\tau_{l}a_{l}}+\frac{1}{z}\sum_{j\neq l}\frac{\tau_{j}}{1+\tau_{j}b_{j}^{l}}||\Sigma^{\frac{1}{2}}Y_{l}|^{2}Cp^{\delta}
\end{align*}
The last inequality is due to Lemma 1 (where $q=\Sigma^{\frac{1}{2}}Y_{l}$,
$B=(S_{l}-zI)^{-1},$$A=\Sigma$, $w_{1}=-\frac{1}{z}\sum_{l=1}^{n}\frac{\tau_{l}}{1+\tau_{l}a_{l}}$,
$w_{2}=-\frac{1}{z}\sum_{j\neq l}\frac{\tau_{j}}{1+\tau_{j}b_{j}^{l}}$), and Assumption 6. \\
Applying Markov Inequality and Lemma 4,
\begin{align*}
P(|Y_{l}'\Sigma Y_{l}-tr(\Sigma)| & \geq p\epsilon)\leq\frac{E|Y_{l}'\Sigma Y_{l}-\mathrm{tr}(\Sigma)|^{2k}}{(p\epsilon)^{2k}}\\
 & \leq\frac{C_{k}p^{k}||\Sigma||^{2k}}{(p\epsilon)^{2k}}\\
 & \leq\frac{CC_{k}p^{k}p^{2\delta k}}{(p\epsilon)^{2k}}
\end{align*}
\textcolor{black}{Choosing $k>\frac{2}{1-2\delta}$}, and using Borel
Cantelli lemma, we
get $\mathrm{max}_{l}|Y_{l}'\Sigma Y_{l}-tr(\Sigma)|/p\rightarrow0$
a.s. As a consequence $\mathrm{max}_{l}|\Sigma^{\frac{1}{2}}Y_{l}|^{2}/p<M$$\forall n>n_{0}$
for some $n_{0}$.

Define, $c=tr(\Sigma^{\frac{1}{2}}(S-zI)^{-1}\Sigma^{\frac{1}{2}})$,
and consider
\[
\mathrm{max}_{l}p^{\delta}|\sum_{l}\frac{\tau_{l}}{1+\tau_{l}a_{l}}-\sum_{j\neq l}\frac{\tau_{j}}{1+\tau_{j}b_{j}^{l}}|
\]

\begin{align*}
 & =\mathrm{max}_{l}p^{\delta}|\sum_{l}\frac{\tau_{l}}{1+\tau_{l}a_{l}}-\sum_{l}\frac{\tau_{l}}{1+\tau_{l}c}+\sum_{l}\frac{\tau_{l}}{1+\tau_{l}c}-\sum_{j\neq l}\frac{\tau_{j}}{1+\tau_{j}c}+\sum_{j\neq l}\frac{\tau_{j}}{1+\tau_{j}c}-\sum_{j\neq l}\frac{\tau_{j}}{1+\tau_{j}b_{j}^{l}}|\\
 & \leq\mathrm{max}_{l}p^{\delta}|\sum_{l}\frac{\tau_{l}}{1+\tau_{l}a_{l}}-\sum_{l}\frac{\tau_{l}}{1+\tau_{l}c}|+\mathrm{max}_{l}p^{\delta}|\sum_{l}\frac{\tau_{l}}{1+\tau_{l}c}-\sum_{j\neq l}\frac{\tau_{j}}{1+\tau_{j}c}|+\mathrm{max}_{l}|\sum_{j\neq l}\frac{\tau_{j}}{1+\tau_{j}c}-\sum_{j\neq l}\frac{\tau_{j}}{1+\tau_{j}b_{j}^{l}}|\\
 & =\mathrm{max}_{l}p^{\delta}|\sum_{l}\frac{\tau_{l}}{1+\tau_{l}a_{l}}-\sum_{l}\frac{\tau_{l}}{1+\tau_{l}c}|+\mathrm{max}_{l}p^{\delta}|\frac{\tau_{l}}{1+\tau_{l}c}|+\mathrm{max}_{l}|\sum_{j\neq l}\frac{\tau_{j}}{1+\tau_{j}c}-\sum_{j\neq l}\frac{\tau_{j}}{1+\tau_{j}b_{j}^{l}}|
\end{align*}

$\mathrm{max}_{l}p^{\delta}|\frac{\tau_{l}}{1+\tau_{l}c}|\leq\mathrm{max}_{l}p^{\delta}|\frac{n\tau_{l}}{n(1+\tau_{l}c)}|\leq\frac{\kappa p^{\delta}}{n}\rightarrow0$

Now we will consider the third part of the above equation,
\begin{align*}
 & \mathrm{max}_{l}p^{\delta}|\sum_{j\neq l}\frac{\tau_{j}}{1+\tau_{j}b_{j}^{l}}-\sum_{j\neq l}\frac{\tau_{j}}{1+\tau_{j}c}|\\
 & =\mathrm{max}_{l}p^{\delta}|\sum_{j\neq l}\frac{\tau_{j}^{2}(c-b_{j}^{l})}{(1+\tau_{j}b_{j}^{l})(1+\tau_{j}c)}|\\
 & =\mathrm{max}_{l}p^{\delta}|\sum_{j\neq l}\frac{\tau_{j}^{2}p(c-b_{j}^{l})/p}{(1+\tau_{j}b_{j}^{l})(1+\tau_{j}c)}|
\end{align*}

$\mathrm{max}_{l}|\sum_{j\neq l}\frac{\tau_{j}^{2}p}{(1+\tau_{j}b_{j}^{l})(1+\tau_{j}c)}|\leq\frac{\kappa^{2}p}{n^{2}}$
by\textcolor{black}{{} Assumption 1 and 6,}

\begin{align*}
 & \mathrm{max}_{l}\mathrm{max}_{j\neq l}p^{\epsilon}|b_{j}^{l}/p-c/p|\\
 & =\mathrm{max}_{l}\mathrm{max}_{j\neq l}p^{\epsilon}|Y_{j}'\Sigma^{\frac{1}{2}}(S_{j,l}-zI)^{-1}\Sigma^{\frac{1}{2}}Y_{j}/p-tr(\Sigma^{\frac{1}{2}}(S-zI)^{-1}\Sigma^{\frac{1}{2}}/p)|\\
 & =\mathrm{max}_{l}\mathrm{max}_{j\neq l}p^{\epsilon}|Y_{j}'\Sigma^{\frac{1}{2}}(S_{j,l}-zI)^{-1}\Sigma^{\frac{1}{2}}Y_{j}/p-tr(\Sigma^{\frac{1}{2}}(S_{j,l}-zI)^{-1}\Sigma^{\frac{1}{2}}/p)|\\
 & +\mathrm{max}_{l}\mathrm{max}_{j\neq l}p^{\epsilon}|tr(\Sigma^{\frac{1}{2}}(S_{j,l}-zI)^{-1}\Sigma^{\frac{1}{2}}/p)-tr(\Sigma^{\frac{1}{2}}(S-zI)^{-1}\Sigma^{\frac{1}{2}}/p)|
\end{align*}

\textcolor{black}{Use of Lemma 4, Lemma 3 with Borel Cantelli Lemma
will give us for} $\epsilon<\frac{1}{2}-\delta$ and $k>\frac{3}{1-2\delta-2\epsilon}$
,
\[
\mathrm{max}_{l}\mathrm{max}_{j\neq l}p^{\epsilon}|Y_{j}'\Sigma^{\frac{1}{2}}(S_{j,l}-zI)^{-1}\Sigma^{\frac{1}{2}}Y_{j}/p-tr(\Sigma^{\frac{1}{2}}(S_{j,l}-zI)^{-1}\Sigma^{\frac{1}{2}}/p)|\rightarrow0\ \mathrm{a.s.}
\]

Also,
\begin{align*}
 & \mathrm{max}_{l}\mathrm{max}_{j\neq l}p^{\epsilon}|\frac{1}{p}tr(\Sigma^{\frac{1}{2}}(S_{j,l}-zI)^{-1}\Sigma^{\frac{1}{2}})-\frac{1}{p}tr(\Sigma^{\frac{1}{2}}(S-zI)^{-1}\Sigma^{\frac{1}{2}})|\\
 & =\mathrm{max}_{l}\mathrm{max}_{j\neq l}p^{\epsilon}|\frac{1}{p}tr[\Sigma^{\frac{1}{2}}\{(S_{j,l}-zI)^{-1}-(S-zI)^{-1}\}\Sigma^{\frac{1}{2}}]|\\
 & =\mathrm{max}_{l}\mathrm{max}_{j\neq l}p^{\epsilon}|\frac{1}{p}tr[\{(S_{j,l}-zI)^{-1}-(S-zI)^{-1}\}\Sigma]|\\
 & \leq\mathrm{max}_{l}\mathrm{max}_{j\leq l}p^{\epsilon}\frac{1}{p}\frac{||\Sigma||}{v}\\
 & \leq\frac{1}{p}\frac{Cp^{\delta+\epsilon}}{v}\rightarrow0\ \mathrm{a.s.}
\end{align*}
The first and second inequalities are result of
application of Lemma 2 and Assumption 6 respectively.\\
This proves out claim.
\end{proof}

\begin{lemma}\label{lemma2app}
For $Y_{l}$, $\Sigma$, $\tau_{l}$ described as Assumption $\mathcal{A}$
and $z=iv$, define,
\[
M_{3}=Y_{l}'\Sigma^{\frac{1}{2}}(S_{l}-zI)^{-1}(\Sigma\sum_{j\neq l}\frac{\tau_{j}}{1+\tau_{j}b_{j}^{l}}-zI)^{-1}\Sigma^{\frac{1}{2}}Y_{l}
\]
and
\[
M_{4}=tr[\Sigma^{\frac{1}{2}}(S_{l}-zI)^{-1}(\Sigma\sum_{j\neq l}\frac{\tau_{j}}{1+\tau_{j}b_{j}^{l}}-zI)^{-1}\Sigma^{\frac{1}{2}}],
\]
then
\[
\frac{1}{p}max_{l}|M_{3}-M_{4}|\rightarrow0\ a.s.
\]
\end{lemma}
\begin{proof}[Proof of Lemma \ref{lemma2app}]
Using Lemma 4 and and Markov Inequality it is easy to show that,
\begin{align*}
 & E(Y_{l}'\Sigma^{\frac{1}{2}}(S_{l}-zI)^{-1}(\Sigma\sum_{j\neq l}\frac{\tau_{j}}{1+\tau_{j}b_{j}^{l}}-zI)^{-1}\Sigma^{\frac{1}{2}}Y_{l}-tr[\Sigma^{\frac{1}{2}}(S_{l}-zI)^{-1}(\Sigma\sum_{j\neq l}\frac{\tau_{j}}{1+\tau_{j}b_{j}^{l}}-zI)^{-1}\Sigma^{\frac{1}{2}}])\\
 & \leq C_{k}p^{k}\frac{p^{2\delta k}}{v^{2k}}
\end{align*}
After choosing appropriate value of $k$ and using Borel Cantelli Lemma
we can get the claim.
\end{proof}
\section{Proof of the Remarks:}
\begin{proof}[Proof of Remark 1]
\begin{align*}
M & =\sqrt{\frac{n}{p}}(WW^{*}-(\sum_{l}\tau_{l})\Lambda)\\
 & =\sqrt{\frac{n}{p}}(\sum_{k=1}^{p}e_{k}w_{k}^{*}W^{*}-(\sum_{l}\tau_{l})\Lambda)\\
 & =\sqrt{\frac{n}{p}}\sum_{k=1}^{p}e_{k}w_{k}^{*}(W_{l}^{*}+w_{k}e_{k}^{*})-\sqrt{\frac{n}{p}}(\sum_{l}\tau_{l})\Lambda\\
 & =\sqrt{\frac{n}{p}}\sum_{k=1}^{p}e_{k}w_{k}^{*}W_{k}^{*}+\sqrt{\frac{n}{p}}\sum_{k=1}^{p}e_{k}w_{k}^{*}w_{k}e_{k}^{*}-\sqrt{\frac{n}{p}}(\sum_{l}\tau_{l})\Lambda\\
 & =\sqrt{\frac{n}{p}}\sum_{k=1}^{p}e_{k}(W_{k}w_{k})^{*}+\sqrt{\frac{n}{p}}\sum_{k=1}^{p}e_{k}w_{k}^{*}w_{k}e_{k}^{*}-\sqrt{\frac{n}{p}}(\sum_{l}\tau_{l})\sum_{k=1}^{p}\lambda_{k}e_{k}e_{k}^{*}\\
 & =\sum_{k=1}^{p}e_{k}h_{k}^{*}+\sum_{k=1}^{p}e_{k}\sqrt{\frac{n}{p}}(w_{k}^{*}w_{k}-(\sum_{l}\tau_{l})\lambda_{k})e_{k}^{*}\\
 & =\sum_{k=1}^{p}e_{k}h_{k}^{*}+\sum_{k=1}^{p}e_{k}\sqrt{\frac{n}{p}}(w_{k}^{*}w_{k}-(\sum_{l}\tau_{l})\lambda_{k})e_{k}^{*}\\
 & =\sum_{k=1}^{p}e_{k}h_{k}^{*}+\sum_{k=1}^{p}e_{k}t_{kk}e_{k}^{*}\\
 & =\sum_{k=1}^{p}e_{k}(h_{k}+t_{kk}e_{k})^{*},
\end{align*}
also it is easy to see that,

\begin{align*}
M & =\sqrt{\frac{n}{p}}(WW^{*}-(\sum_{l}\tau_{l})\Lambda)\\
 & =\sqrt{\frac{n}{p}}(W_{k}W^{*}+e_{k}w_{k}^{*}W^{*}-(\sum_{l}\tau_{l})\Lambda_{k}-(\sum_{l}\tau_{l})e_{k}e_{k}^{*}\lambda_{k})\\
 & =\sqrt{\frac{n}{p}}(W_{k}W^{*}-(\sum_{l}\tau_{l})\Lambda_{k})+\sqrt{\frac{n}{p}}[e_{k}w_{k}^{*}(W_{k}^{*}+w_{k}e_{k}^{*})-(\sum_{l}\tau_{l})e_{k}e_{k}^{*}\lambda_{k}]\\
 & =M_{k}+e_{k}h_{k}^{*}+\sqrt{\frac{n}{p}}e_{k}w_{k}^{*}w_{k}e_{k}^{*}-\sqrt{\frac{n}{p}}(\sum_{l}\tau_{l})e_{k}e_{k}^{*}\lambda_{k}\\
 & =M_{k}+e_{k}h_{k}^{*}+e_{k}\sqrt{\frac{n}{p}}(w_{k}^{*}w_{k}-(\sum_{l}\tau_{l})\lambda_{k})e_{k}^{*}\\
 & =M_{k}+e_{k}h_{k}^{*}+t_{kk}e_{k}e_{k}^{*}\\
 & =M_{k}+e_{k}(h_{k}+t_{kk}e_{k})^{*}.
\end{align*}
\end{proof}

\begin{proof}[Proof of Remark2:]
\begin{align*}
 & tr[e_{k}(h_{k}+t_{kk}e_{k})^{*}(M_{k}-zI)^{-1}(h_{k}+t_{kk}e_{k})^{*}(M_{k}-zI)^{-1}e_{k}]\\
 & =(h_{k}+t_{kk}e_{k})^{*}(M_{k}-zI)^{-1}e_{k}tr[e_{k}(h_{k}+t_{kk}e_{k})^{*}(M_{k}-zI)^{-1}]\\
 & =\{(h_{k}+t_{kk}e_{k})^{*}(M_{k}-zI)^{-1}e_{k}\}\{(h_{k}+t_{kk}e_{k})^{*}(M_{k}-zI)^{-1}e_{k}\}\\
 & =\{(h_{k}+t_{kk}e_{k})^{*}(M_{k}-zI)^{-1}e_{k}\}^{2}
\end{align*}
and
\begin{align*}
 & tr[e_{k}(h_{k}+t_{kk}e_{k})^{*}(M_{k}-zI)^{-1}e_{k}(h_{k}+t_{kk}e_{k})^{*}(M_{k}-zI)^{-1}]\\
 & =tr[(h_{k}+t_{kk}e_{k})^{*}(M_{k}-zI)^{-1}e_{k}(h_{k}+t_{kk}e_{k})^{*}(M_{k}-zI)^{-1}e_{k}]\\
 & =\{(h_{k}+t_{kk}e_{k})^{*}(M_{k}-zI)^{-1}e_{k}\}^{2}.
\end{align*}
So,
\begin{align*}
LHS & =tr(M(M-zI)^{-1})\\
 & =tr\{e_{k}(h_{k}+t_{kk}e_{k})^{*}(M_{k}-zI+e_{k}(h_{k}+t_{kk}e_{k})^{*})^{-1}\}\\
 & =tr[e_{k}(h_{k}+t_{kk}e_{k})^{*}(M_{k}(z)^{-1}-\frac{(M_{k}-zI)^{-1}e_{k}(h_{k}+t_{kk}e_{k})^{*}(M_{k}-zI)^{-1}}{1+(h_{k}+t_{kk}e_{k})^{*}(M_{k}-zI)^{-1}e_{k}})]\\
 & =tr[e_{k}(h_{k}+t_{kk}e_{k})^{*}((M_{k}-zI)^{-1}-\frac{(M_{k}-zI)^{-1}e_{k}(h_{k}+t_{kk}e_{k})^{*}(M_{k}-zI)^{-1}}{1+(h_{k}+t_{kk}e_{k})^{*}(M_{k}-zI)^{-1}e_{k}})]\\
 & =tr[e_{k}(h_{k}+t_{kk}e_{k})^{*}(\frac{(M_{k}-zI)^{-1}(1+(h_{k}+t_{kk}e_{k})^{*}(M_{k}-zI)^{-1}e_{k})-(M_{k}-zI)^{-1}e_{k}(h_{k}+t_{kk}e_{k})^{*}(M_{k}-zI)^{-1}}{1+(h_{k}+t_{kk}e_{k})^{*}M_{k}(z)^{-1}e_{k}})\\
 & =tr[\frac{e_{k}(h_{k}+t_{kk}e_{k})^{*}(M_{k}-zI)^{-1}}{1+(h_{k}+t_{kk}e_{k})^{*}(M_{k}-zI)^{-1}e_{k}}]\\
 & =\frac{(h_{k}+t_{kk}e_{k})^{*}(M_{k}-zI)^{-1}e_{k}}{1+(h_{k}+t_{kk}e_{k})^{*}(M_{k}-zI)^{-1}e_{k}}
\end{align*}
\end{proof}

\begin{proof}[Proof of Remark 3]
To see this, manipulate the left hand side the following way-

\begin{align*}
(h_{k}+t_{kk}e_{k})^{*}(M_{k}-zI)^{-1}e_{k} & =(h_{k}+t_{kk}e_{k})^{*}(\sqrt{\frac{n}{p}}(W_{k}W^{*}-(\sum_{l}\tau_{l})\Lambda_{k})-zI)^{-1}e_{k}\\
 & =(h_{k}+t_{kk}e_{k})^{*}(\sqrt{\frac{n}{p}}(W_{k}W_{k}^{*}+W_{k}w_{k}e_{k}^{*}-(\sum_{l}\tau_{l})\Lambda_{k})-zI)^{-1}e_{k}\\
 & =(h_{k}+t_{kk}e_{k})^{*}(\sqrt{\frac{n}{p}}(W_{k}W_{k}^{*}-(\sum_{l}\tau_{l})\Lambda_{k})-zI+\sqrt{\frac{n}{p}}W_{k}w_{k}e_{k}^{*})^{-1}e_{k}\\
 & =(h_{k}+t_{kk}e_{k})^{*}(\sqrt{\frac{n}{p}}(W_{k}W_{k}^{*}-(\sum_{l}\tau_{l})\Lambda_{k})-zI+h_{k}e_{k}^{*})^{-1}e_{k}\\
 & =(h_{k}+t_{kk}e_{k})^{*}((\bar{M}_{k}-zI)+h_{k}e_{k}^{*})^{-1}e_{k}\\
 & =(h_{k}+t_{kk}e_{k})^{*}(\bar{(M}_{k}-zI)^{-1}-\frac{\bar{(M}_{k}-zI)^{-1}h_{k}e_{k}^{*}\bar{(M}_{k}-zI)^{-1}}{1+e_{k}^{*}\bar{(M}_{k}-zI)^{-1}h_{k}})e_{k}\\
 & =(h_{k}+t_{kk}e_{k})^{*}\bar{M}_{k}(z)^{-1}e_{k}-\frac{(h_{k}+t_{kk}e_{k})^{*}\bar{(M}_{k}-zI)^{-1}h_{k}e_{k}^{*}\bar{(M}_{k}-zI)^{-1}e_{k}}{1+e_{k}^{*}\bar{(M}_{k}-zI)^{-1}h_{k}}.
\end{align*}
Observe that $h_{k}^{*}e_{k}=e^{*}h_{k}=0$ because $W_{k}^{*}e_{k}=0$.
So,
\begin{align*}
\bar{(M}_{k}-zI)e_{k} & =\sqrt{\frac{n}{p}}(W_{k}W_{k}^{*}-(\sum_{l}\tau_{l})\Lambda_{k})e_{k}-ze_{k}\\
 & =\sqrt{\frac{n}{p}}(W_{k}W_{k}^{*}e_{k}-(\sum_{l}\tau_{l})\Lambda_{k}e_{k})-ze_{k}\\
 & =-ze_{k}.
\end{align*}
As a consequence $e_{k}^{*}(\bar{M}_{k}-zI)^{-1}h_{k}=-e_{k}^{*}h_{k}/z=0$.
And,
\begin{align*}
(h_{k}+t_{kk}e_{k})^{*}(M_{k}-zI)^{-1}e_{k} & =-(h_{k}+t_{kk}e_{k})^{*}e_{k}/z+(h_{k}+t_{kk}e_{k})^{*}\bar{(M}_{k}-zI)^{-1}h_{k}e_{k}^{*}e_{k}/z\\
 & =-t_{kk}/z+(h_{k}+t_{kk}e_{k})^{*}\bar{(M}_{k}-zI)^{-1}h_{k}/z\\
 & =\frac{-t_{kk}+h_{k}^{*}\bar{(M}_{k}-zI)^{-1}h_{k}+t_{kk}e_{k}{}^{*}\bar{(M}_{k}-zI)^{-1}h_{k}}{z}\\
 & =\frac{-t_{kk}+h_{k}^{*}\bar{(M}_{k}-zI)^{-1}h_{k}}{z}.
\end{align*}
\end{proof}

\begin{proof}[Proof of Remark 4]
\begin{align*}
h_{k}^{*}\bar{(M}_{k}-zI)^{-1}h_{k} & =\frac{n}{p}w_{k}^{*}W_{k}^{*}\bar{(M}_{k}-zI)^{-1}W_{k}w_{k}\\
 & =\frac{n}{p}tr(W_{k}w_{k}w_{k}^{*}W_{k}^{*}\bar{(M}_{k}-zI)^{-1})\\
 & =\frac{n\lambda_{k}}{pn}tr(W_{k}diag(\tau_{1}^{\frac{1}{2}},..,\tau_{n}^{\frac{1}{2}})UY_{k}Y_{k}^{*}U^{*}diag(\tau_{1}^{\frac{1}{2}},..,\tau_{n}^{\frac{1}{2}})W_{k}^{*}\bar{(M}_{k}-zI)^{-1})\\
 & =\frac{n\lambda_{k}}{p}tr(W_{k}diag(\tau_{1}^{\frac{1}{2}},..,\tau_{n}^{\frac{1}{2}})U(Y_{k}Y_{k}^{*}-I)U^{*}diag(\tau_{1}^{\frac{1}{2}},..,\tau_{n}^{\frac{1}{2}})W_{k}^{*}\bar{(M}_{k}-zI)^{-1})\\
 & +\frac{n\lambda_{k}}{p}tr(W_{k}diag(\tau_{1},..,\tau_{n})W_{k}^{*}\bar{(M}_{k}-zI)^{-1})\\
 & =\frac{n\lambda_{k}}{p}tr(W_{k}diag(\tau_{1}^{\frac{1}{2}},..,\tau_{n}^{\frac{1}{2}})U(Y_{k}Y_{k}^{*}-I)U^{*}diag(\tau_{1}^{\frac{1}{2}},..,\tau_{n}^{\frac{1}{2}})W_{k}^{*}\bar{(M}_{k}-zI)^{-1})\\
 & +\frac{n\lambda_{k}}{p}tr(\Lambda_{k}^{\frac{1}{2}}UY_{n}(diag(\tau_{1},..,\tau_{n}))^{2}Y_{n}^{*}U^{*}\Lambda_{k}^{\frac{1}{2}}\bar{(M}_{k}-zI)^{-1})\\
 & =\frac{n\lambda_{k}}{p}tr(W_{k}diag(\tau_{1}^{\frac{1}{2}},..,\tau_{n}^{\frac{1}{2}})U(Y_{k}Y_{k}^{*}-I)U^{*}diag(\tau_{1}^{\frac{1}{2}},..,\tau_{n}^{\frac{1}{2}})W_{k}^{*}\bar{(M}_{k}-zI)^{-1})\\
 & +\frac{n\lambda_{k}}{p}tr(\Lambda_{k}^{\frac{1}{2}}UY_{n}(\sum_{l}\tau_{n}^{2}e_{l}e_{l}^{*}-(\sum_{l}\tau_{n}^{2})I)Y_{n}^{*}U^{*}\Lambda_{k}^{\frac{1}{2}}\bar{(M}_{k}-zI)^{-1})\\
 & +\frac{n\lambda_{k}}{p}tr(\Lambda_{k}^{\frac{1}{2}}UY_{n}(\sum_{l}\tau_{n}^{2})Y_{n}^{*}U^{*}\Lambda_{k}^{\frac{1}{2}}\bar{(M}_{k}-zI)^{-1})\\
 & =\frac{n\lambda_{k}}{p}tr(W_{k}diag(\tau_{1}^{\frac{1}{2}},..,\tau_{n}^{\frac{1}{2}})U(Y_{k}Y_{k}^{*}-I)U^{*}diag(\tau_{1}^{\frac{1}{2}},..,\tau_{n}^{\frac{1}{2}})W_{k}^{*}\bar{(M}_{k}-zI)^{-1})\\
 & +\frac{n\lambda_{k}}{p}tr(\Lambda_{k}^{\frac{1}{2}}UX_{n}(\sum_{l}\tau_{n}^{2}e_{l}e_{l}^{*}-(\sum_{l}\tau_{n}^{2})I)Y_{n}^{*}U^{*}\Lambda_{k}^{\frac{1}{2}}\bar{(M}_{k}-zI)^{-1})\\
 & +\frac{n\lambda_{k}(\sum_{l}\tau_{n}^{2})}{p}tr(\Lambda_{k}^{\frac{1}{2}}U(Y_{n}Y_{n}^{*}-I)U^{*}\Lambda_{k}^{\frac{1}{2}}\bar{(M}_{k}-zI)^{-1})\\
 & +\frac{n\lambda_{k}(\sum_{l}\tau_{n}^{2})}{p}tr(\Lambda_{k}\bar{(M}_{k}-zI)^{-1}).
\end{align*}
\end{proof}

\end{document}